\documentclass[twoside,leqno]{article}
\usepackage[letterpaper]{geometry}
\usepackage{ltexpprt}
\usepackage{hyperref}

\newif\ifArxivVersion
\ArxivVersiontrue

\usepackage{etoolbox}
\undef{\theorem}\undef{\endtheorem}\undef{\thetheorem}
\newtheorem{theorem}{Theorem}[section]
\undef{\lemma}\undef{\endlemma}\undef{\thelemma}
\newtheorem{lemma}[theorem]{Lemma}
\undef{\Definition}\undef{\endDefinition}\undef{\theDefinition}
\newtheorem{Definition}[theorem]{Definition}
\undef{\algorithm}\undef{\endalgorithm}\undef{\thealgorithm}
\newtheorem{algorithm}[theorem]{Algorithm}

\usepackage{iftex}
\ifpdftex
\usepackage[utf8]{inputenc}
\usepackage[T1]{fontenc}
\fi

\usepackage{enumitem}
\usepackage{amsmath}
\usepackage{amsfonts}

\usepackage{xspace}
\usepackage{tikz}
    \usetikzlibrary{patterns,patterns.meta,calc,math}
\usepackage{mathtools}

\usepackage{thm-restate}
\usepackage{etoolbox,xpatch}
\makeatletter
\xpatchcmd\thmt@restatable{%
\csname #2\@xa\endcsname\ifx\@nx#1\@nx\else[{#1}]\fi
}{%
\ifthmt@thisistheone
\csname #2\@xa\endcsname\ifx\@nx#1\@nx\else[{#1}]\fi
\else
\csname #2\@xa\endcsname[{restated}]
\fi}{}{}
\makeatother

\definecolor{lightblue}{HTML}{8D9DB6}

\newcommand{\pp}[1]{\textup{#1}}

\newcommand{\PESqd}[2]{\ensuremath{\pp{PES}_{#1}^{#2}}}
\newcommand{\Roots}[2]{\ensuremath{\pp{\#Roots}_{#1}^{#2}}}
\newcommand{\SUMqd}[2]{\ensuremath{\pp{Sum}_{#1}^{#2}}}

\usepackage{algpseudocodex}
\renewcommand{\Comment}[1]{\textit{\hypersetup{linkcolor=gray}\textcolor{gray}{#1}}}
\renewcommand{\Call}[2]{\textup{\textsc{#1}%
\if\relax\detokenize{#2}\relax\else\ensuremath{(#2)}\fi}}
\newcommand{\RazborovSmolensky}[1]{\Call{\hyperref[algo:RazborovSmolensky]{RazborovSmolensky}}{#1}}
\newcommand{\PartialSum}[1]{\Call{\hyperref[algo:PartialSum]{PartialSum}}{#1}}
\newcommand{\FullSum}[1]{\Call{\hyperref[algo:FullSum]{FullSum}}{#1}}

\usepackage{framed}
\newlength{\RoundedBoxWidth}
\newsavebox{\GrayRoundedBox}
\newenvironment{GrayBox}[1]%
   {\setlength{\RoundedBoxWidth}{.93\textwidth}
    \def\boxheading{#1}
    \begin{lrbox}{\GrayRoundedBox}
       \begin{minipage}{\RoundedBoxWidth}}%
   {   \end{minipage}
    \end{lrbox}
    \begin{center}
    \begin{tikzpicture}%
       \node(Text)[draw=black!20,rounded corners,%
             inner sep=0.7em,text width=\RoundedBoxWidth,anchor=north west]
             at (0,0) {\usebox{\GrayRoundedBox}};
        \node [rectangle,inner sep=0.3em,anchor=base west,fill=white] at (0.8em,-0.1em) {\boxheading};
    \end{tikzpicture}
    \end{center}}

\usepackage{tabularx}
\newcommand{\defproblemtalt}[3]{
  \begin{GrayBox}{#1}
    \begin{tabular*}{\textwidth}{>{\itshape}l p{0.85\textwidth}}
    Input:  & #2 \\
    Question: & #3
    \end{tabular*}
  \end{GrayBox}
}%

\newcommand{\E}{\mathbb{E}}
\newcommand{\F}{\mathbb{F}}
\newcommand{\R}{\mathbb{R}}

\newcommand{\N}{\mathbb{N}}

\newcommand{\OO}{\mathcal{O}}

\newcommand{\OOstar}{\OO^\ast}

\newcommand{\eps}{\varepsilon}
\renewcommand{\epsilon}{\eps}
\newcommand{\zo}{\set{0,1}}

\DeclarePairedDelimiter{\ceil}{\lceil}{\rceil}
\DeclarePairedDelimiter{\floor}{\lfloor}{\rfloor}
\DeclarePairedDelimiter\paren{\lparen}{\rparen}  
\DeclarePairedDelimiter\abs{\lvert}{\rvert}  
\DeclarePairedDelimiter\set{\{}{\}}  
\DeclarePairedDelimiterX\setc[2]{\{}{\}}{\,#1 \mid #2\,}  
\DeclarePairedDelimiterX\parenc[2]{\lparen}{\rparen}{\,#1 \;\delimsize\vert\; #2\,}  

\DeclareMathOperator{\poly}{poly}

\newcommand{\FF}{\varphi}

\newcommand{\PartialSumError}[1]{q^{- #1}}
\newcommand{\field}{\ensuremath{\F_q}\xspace}
\newcommand{\TrimmedFp}[2]{\F_{q\downarrow #2}^{#1}}
\newcommand{\Zdegree}[2]{\Delta_{#1,#2}}

\newcommand{\kappaUPPER}{\ensuremath{\tfrac{1}{2d-1}}}

\begin{document}

\newcommand\relatedversion{}

\ifArxivVersion
\renewcommand\relatedversion{\thanks{An extended abstract of this work appears in the Proceedings of the Symposium on Discrete Algorithms (SODA 2025), published by the
Society for Industrial and Applied Mathematics (SIAM).}}
\fi

\title{\Large Solving Polynomial Equations Over Finite Fields\relatedversion}

\author{%
Holger Dell%
\thanks{Goethe University Frankfurt, Germany; IT University of Copenhagen and Basic Algorithms Research Copenhagen (BARC), Denmark.}%
\and Anselm Haak%
\thanks{Universität Paderborn, Germany.}%
\and Melvin Kallmayer%
\thanks{Goethe University Frankfurt, Germany.}%
\and Leo Wennmann%
\thanks{Maastricht University, Netherlands. Supported by Dutch Research Council (NWO) project ``The Twilight Zone of Efficiency: Optimality of Quasi-Polynomial Time Algorithms'' [grant number OCEN.W.21.268].}%
}


\date{}

\maketitle




\ifArxivVersion
\else
\fancyfoot[R]{\scriptsize{Copyright \textcopyright\ 2025\\
Copyright for this paper is retained by authors}}
\fi



\begin{abstract}\small\baselineskip=9pt%
    We present a randomized algorithm for solving low-degree polynomial equation systems over finite fields faster than exhaustive search.
    In order to do so, we follow a line of work by Lokshtanov, Paturi, Tamaki, Williams, and Yu (SODA 2017), Björklund, Kaski, and Williams (ICALP 2019), and Dinur (SODA 2021).
    In particular, we generalize Dinur's algorithm for $\F_2$ to all finite fields, in particular the ``symbolic interpolation'' of Björklund, Kaski, and Williams, and we use an efficient trimmed multipoint evaluation and interpolation procedure for multivariate polynomials over finite fields by Van der Hoeven and Schost (AAECC 2013).
    The running time of our algorithm matches that of Dinur's algorithm for $\F_2$ and is significantly faster than the one of Lokshtanov et al. for $q>2$.

    We complement our results with tight conditional lower bounds that, surprisingly, we were not able to find in the literature. In particular, under the strong exponential time hypothesis, we prove that it is impossible to solve $n$-variate low-degree polynomial equation systems over $\field$ in time~$\OO((q-\eps)^{n})$.
    As a bonus, we show that under the counting version of the strong exponential time hypothesis, it is impossible to compute the number of roots of a \emph{single} $n$-variate low-degree polynomial over $\field$ in time~${\OO((q-\eps)^{n})}$; this generalizes a result of Williams (SOSA 2018) from $\F_2$ to all finite fields.
    \nocite{lokshtanov2017beating}%
    \nocite{DBLP:conf/icalp/BjorklundK019}%
    \nocite{dinur2021improved}%
    \nocite{DBLP:journals/aaecc/HoevenS13}%
    \nocite{DBLP:conf/soda/Williams18a}%
\end{abstract}

\section{Introduction}

In the 16th century, mathematicians have kept algorithms for solving polynomial equation systems secret from each other, so that they could publicly demonstrate their superior skill in case they were challenged for their non-tenured positions (e.g., \cite{veritasium,toscano2020secret,Rowe_2014}).
Modern scientists are still facing similar problems, except that secrecy is frowned upon, and so we chose to publish our algorithm for polynomial equation systems over finite fields in this paper.

Polynomial equation systems over finite fields have countless applications.
In coding theory, they are used in the decoding step of error-correcting codes~\cite{pellikaan_wu_bulygin_jurrius_2017}.
In cryptography, they can be used to break cryptographic schemes~\cite{pellikaan_wu_bulygin_jurrius_2017}.
They even have a close connection to quantum computations, where the output of such computations can be expressed as the number of solutions to a set of polynomial equations, allowing for elegant proofs of some relations between quantum and classical complexity classes~\cite{DBLP:journals/qic/DawsonHMHNO05}.

\subsection{Our Results.}
For every prime power~$q$ and every positive integer~$d$, we define the problem of solving a system of degree-$d$ polynomial equations over the finite field~$\field$ as follows:

\defproblemtalt{\PESqd{q}{d}}{Polynomials $P_1, \dots, P_m \in \field[X_1,\dots,X_n]$ of degree at most~$d$}{Is there an assignment~$x\in\field^n$ with $P_1(x)=P_2(x)=\dots=P_m(x)=0$?}

We stress that the polynomials are explicitly given as a list of monomial-coefficient pairs.
Exhaustive search trivially solves this problem in time $\OOstar(q^{n})$ asymptotically in $n$ for any fixed $q$ and $d$, where the $\OOstar(\cdot)$-notation hides polynomial factors.
Our main result is a faster algorithm for this problem.
\begin{restatable}{theorem}{PESalgorithm}\label{thm:pes-algorithm}
    For every prime power~$q$ and every positive integer~$d$, there is a bounded-error randomized algorithm that solves \PESqd{q}{d} in time $\OOstar(q^{\zeta_{q,d} n})$ for  $0<\zeta_{q,d}\leq 1 - \min\paren*{\tfrac{1}{8\ln q}, \tfrac{1}{4d}}$.
\end{restatable}
To emphasize, $q$ and $d$ are treated as constants for this algorithm, and its running time is given asymptotically in $n$.
The exponent~$\zeta_{q,d}$ is a non-elementary function of~$q$ and~$d$ that we explicitly define in Section~\ref{sec:running-time}, and we do not believe our upper bound to be tight.
For example, numerical experiments suggest $\zeta_{2,d}\leq 1-1/(2d)$ for all $d\leq 2^{18}$.
We provide a small table of running times for specific values of $q$ and $d$ in Figure~\ref{fig:running-times}.
For $q=2$, our algorithm specializes to the algorithm of Dinur~\cite{dinur2021improved} that we build on, and for $q>2$, the running time we obtain is significantly faster than the one obtained by Lokshtanov et al.~\cite{lokshtanov2017beating}.

\begin{figure}[ht]
    \centering
\begin{tabular}{lcccc}
    Algorithm & $q=d=2$ & $q=3$, $d=2$ & $q=4$, $d=2$ & $q=4$, $d=3$ \\\hline
    Lokshtanov et al.~\cite{lokshtanov2017beating} & $\OOstar(2^{0.8765n})$
    & $\OOstar(q^{0.9975n})$
    & $\OOstar(q^{0.9n})$
    & $\OOstar(q^{0.934n})$
    \\
    Björklund, Kaski, Williams~\cite{DBLP:conf/icalp/BjorklundK019} & $\OOstar(2^{0.804n})$ & --- & --- & --- \\
    Dinur~\cite{dinur2021improved} & \boldmath$\OOstar(2^{0.695n})$ & --- & --- & --- \\
    This paper
    & \boldmath$\OOstar(2^{0.695n})$
    & \boldmath$\OOstar(q^{0.696n})$
    & \boldmath$\OOstar(q^{0.698n})$
    & \boldmath$\OOstar(q^{0.813n})$\\
\end{tabular}
\caption{\label{fig:running-times}Comparison of running times for some values of $q$ and $d$.}
\end{figure}

When $q$ remains constant, the term~$\zeta_{q,d}$ in the running time of our algorithm approaches~$1$ as the degree~$d$ tends to infinity.
We show this behavior is necessary under the Strong Exponential Time Hypothesis (\pp{SETH}).
\begin{restatable}{theorem}{PEShardness}\label{thm:pes-hardness}\mbox{}
    If \pp{SETH} holds, then for all prime powers $q$ and all rationals $\delta >0$, there exists $d\in\N$ such that there is no $\OOstar(q^{(1-\delta) n})$-time algorithm for \PESqd{q}{d}.
\end{restatable}
We leave the reverse setting as an open question: If $d$ remains constant (say, $d=2$), does there exist a constant $\delta>0$ such that \PESqd{q}{d} can be solved in time $\OOstar(q^{(1-\delta) n})$ for all prime powers~$q$?

As a bonus, we also consider the counting complexity of the problem.
We write \Roots{q}{d} for the counting problem where the goal is to compute the number of roots of a single given $n$-variate degree-$d$ polynomial over~$\field$.
Under the counting version of \pp{SETH}, we obtain the following result.
\begin{restatable}{theorem}{ROOTShardness}\label{thm:roots-hardness}\mbox{}
    If \#\pp{SETH} holds, then for all prime powers $q$ and all rationals $\delta >0$, there exists $d\in\N$ such that there is no $\OO(q^{(1-\delta) n})$-time algorithm for \Roots{q}{d}.
\end{restatable}
We remark that assuming \#\pp{SETH} instead of \pp{SETH} makes the theorem stronger.
The special case~$q=2$ of Theorem~\ref{thm:roots-hardness} was proved by Williams~\cite[Theorem~4]{DBLP:conf/soda/Williams18a}; we combine this proof with our hardness reduction from Theorem~\ref{thm:pes-hardness} to establish the general case $q>2$.

\subsection{Related Work.}

Our main result continues a recent line of work~\cite{lokshtanov2017beating, DBLP:conf/icalp/BjorklundK019, dinur2021improved} on solving systems of polynomial equations over finite fields in time $\OOstar(q^{\delta n})$ for some fixed $\delta < 1$.
Using the polynomial method~\cite{DBLP:conf/coco/Beigel93}, Lokshtanov et al.~\cite{lokshtanov2017beating} obtained a randomized worst-case algorithm better than exhaustive search for any prime power $q$ and degree $d\ge 1$.
They obtain running times of the form $\OOstar(q^{\delta n})$, where $\delta=1-\tfrac{1}{Cd}$ holds for some fixed constant~$C$.
For $q = d = 2$, their algorithm yields a running time of $\OOstar(2^{0.8765n})$ which was improved to $\OOstar(2^{0.804n})$ by Björklund, Kaski, and Williams~\cite{DBLP:conf/icalp/BjorklundK019}, who solely focused on solving \PESqd{2}{d}.
The current state-of-the-art algorithm for \PESqd{2}{d} was devised by Dinur~\cite{dinur2021improved}, significantly improving the running time to $\OOstar(2^{0.695n})$.
All aforementioned algorithms are closely related to our work and will be explained in more detail in Section~\ref{sec:technical-overview}.

We also provide some examples for other settings considered for this problem in the literature and results in those settings, highlighting the diversity of relevant settings, approaches, and applications regarding polynomial equation systems.
As the problem is hard in general (see~Theorem~\ref{thm:pes-hardness}), many approaches aim to solve the problem ``fast enough'' in practice, for example in order to utilize them in cryptographic attacks.
Bard~\cite{DBLP:phd/basesearch/Bard07} and Courtois and Bard~\cite{DBLP:conf/ima/CourtoisB07} solve multivariate polynomial systems over $\mathbb{F}_2$ by reducing them to the satisfiability problem for Boolean formulas.
Bard~\cite{DBLP:phd/basesearch/Bard07} also gives a fast method for solving linear systems of equations over $\mathbb{F}_2$ in practice, combining the so-called Method of Four Russians with Strassen's algorithm.
Ding et al.~\cite{DBLP:journals/iacr/DingGS06} solve multivariate polynomial systems over arbitrary finite fields by reducing the problem to that of solving a single univariate polynomial equation over an extension of the original field.
There is a variety of algorithms based on Gröbner bases, prominent examples being the $F_4$- and the $F_5$-algorithm~\cite{faugere1999,faugere2002} as well as the XL-algorithm~\cite{eurocrypt-2000-2187}.
In many cases, these algorithms output a representation of the set of all solutions.
They can be efficient in restricted cases or at least heuristically, but have very high worst case complexity.
Note that the complexity of the XL-algorithm was not conclusively determined in the original work, see~\cite{abdelmageedMohamed11} for an overview of what is known.

There is also work on restricted cases with a better worst-case complexity than the general problem.
Ivanyos and Santha~\cite{DBLP:journals/tcs/IvanyosS17} gave a polynomial-time algorithm in a restricted setting with applications in quantum computing.
Another line of research aims to obtain complexity bounds for Las Vegas type algorithms for solving polynomial equation systems under certain additional assumptions that are likely to hold for systems with random coefficients.
Here, the goal is to compute a representation of all solutions in a time that is between quadratic and cubic in the number of solutions.
These algorithms typically do not offer efficient ways to answer the corresponding decision problem \PESqd{q}{d}, and they require additional assumptions on the system.
A recent work in this vein is due to van der Hoeven and Lecerf~\cite{DBLP:journals/focm/HoevenL21}, which also provides a good exposition of related work.

\subsection{Our Contributions.}\label{sec:our-contributions}

We briefly mention the key contributions of this paper compared to previous work.
While many individual techniques generalize easily from $\F_2$ to $\field$ to obtain Theorem~\ref{thm:pes-algorithm}, the overall argument is quite delicate.
The following are the most important changes compared to Dinur's algorithm~\cite{dinur2021improved} for $\F_2$:
\begin{itemize}[noitemsep]
    \item Instead of the Fast Möbius Transform, which only works over $\F_2$, we use a suitable Fast Multipoint Evaluation algorithm by Van der Hoeven and Schost~\cite{DBLP:journals/aaecc/HoevenS13} (see Lemma~\ref{lem:interpolation}).
    \item
    Björklund, Kaski, and Williams~\cite{DBLP:conf/icalp/BjorklundK019} as well as Dinur~\cite{dinur2021improved} use
    a cancellation trick modulo~$2$~\cite[Eq.~(19)]{DBLP:conf/icalp/BjorklundK019} to significantly reduce the running time compared to Lokshtanov et al.~\cite{lokshtanov2017beating}. We use Fermat's little theorem to generalize this trick to \emph{symbolic interpolation} over any finite field~$\field$ (see Lemma~\ref{lem:modqtrick}).
    \item Finally, we need to analyze the running time of the algorithm, which turns out to be surprisingly non-trivial in the general case. To do so, we use a bound on the \emph{extended binomial coefficient}, which is equal to the number of degree-$\Delta$ monomials in $\field[X_1,\dots,X_n]$. Entropy-style bounds on the extended binomial coefficient have only recently been studied in mathematics in the context of the cap set problem (see Sections~\ref{sec:approx-number-monomials} and~\ref{sec:running-time}).
\end{itemize}
Moreover, our hardness results (Theorems~\ref{thm:pes-hardness} and~\ref{thm:roots-hardness}) appear to be novel. We follow standard techniques from fine-grained complexity to design a suitable reduction from \pp{$k$-SAT}.

\subsection{Technical Overview.}\label{sec:technical-overview}

We provide a high-level overview of the proof of Theorem~\ref{thm:pes-algorithm} and discuss the key techniques of our algorithm.
We follow the main insights of the previous papers~\cite{lokshtanov2017beating,DBLP:conf/icalp/BjorklundK019,dinur2021improved}.
Although two of these papers~\cite{DBLP:conf/icalp/BjorklundK019,dinur2021improved} only consider the special case of $\F_2$, many techniques generalize to $\field$ for $q\geq2$, which is why we describe the techniques in the general case and sprinkle in our contributions to make the generalization work.

\paragraph*{Polynomial Method using Indicator Polynomial.}
Based on the so-called \emph{polynomial method} from Boolean circuit complexity~\cite{DBLP:conf/coco/Beigel93}, Lokshtanov et al.~\cite{lokshtanov2017beating} devised a randomized algorithm for \PESqd{q}{d} by capturing the entire system of polynomial equations as a single \emph{probabilistic} polynomial.
This system has a \emph{small} exponential number of monomials and has a very high probability of being consistent with the original system on many assignments.
For polynomials $P_1, \dots, P_m \in \field[X_1,\dots,X_n]$, we define the associated \emph{indicator polynomial} $F\in\field[X_1,\dots,X_n]$ as
\[
    F \coloneqq \prod_{i=1}^m \paren*{1-(P_i)^{q-1}}.
\]
By Fermat's little theorem, $F(x) = 1$ holds for a vector~$x\in\field^{n}$ if and only if $P_1(x)=\dots=P_m(x)=0$ holds.
Thus, evaluating $F$ on all points $x \in \field^{n}$ suffices to decide whether the system has a solution, that is, whether the polynomials have a common root.
For more details, refer to Section~\ref{sec:equations-to-sum}.
Unfortunately, this idea alone does not yet yield an efficient algorithm for \PESqd{q}{d}, as we are unable to expand~$F$ as a sum of monomials or evaluate it faster than in time $\OOstar(q^n)$---this is because~$F$ would need to be evaluated at each of the $q^{n}$ evaluation points~$x\in\field^{n}$.

\paragraph*{Probabilistic Polynomial Construction of Razborov--Smolensky.}

Building on the observation in the previous paragraph, Lokshtanov et al.~\cite{lokshtanov2017beating} used the probabilistic polynomial construction of \emph{Razborov and Smolensky}~\cite{razborov1987lower,smolensky1987algebraic}, formally stated in Lemma~\ref{lem:RazborovSmolensky}, to reduce the \emph{number} of polynomials~$P_i$.
It transforms the original polynomials into much fewer, probabilistic polynomials~$\widetilde P_i\in\field[X_1,\dots,X_n]$.
In turn, the indicator polynomial~$\widetilde F\in\field[X_1,\dots,X_n]$ constructed from the polynomials~$\widetilde P_i$ is of smaller degree and each evaluation~$\widetilde P_i(x)$ is likely to be equal to the corresponding evaluation~$F(x)$.
Efficiently evaluating the probabilistically obtained indicator polynomial~$\widetilde F$ on many carefully chosen assignments leads to exponential savings over exhaustive search.

\paragraph*{Randomized Isolation Technique of Valiant--Vazirani.}

Björklund et al.~\cite{DBLP:conf/icalp/BjorklundK019} observed that the previous approach of Lokshtanov et al.~\cite{lokshtanov2017beating} included a \emph{decision-to-parity reduction} within the algorithm which can also be done on the system of polynomials itself using randomized isolation techniques.
One elegant technique is \emph{Valiant--Vazirani affine hashing} \cite{valiant1985np}, formally stated in Lemma~\ref{lem:ValiantVazirani}, that probabilistically transforms a polynomial equation system into one that preserves \emph{exactly one solution} of the original system with high probability, if the latter has a solution, and does not add any solutions otherwise.
The isolation technique probabilistically reduces the decision problem \PESqd{q}{d} to the intermediate problem \SUMqd{q}{d} that computes the \emph{full sum}~$Z \in \field$ over all evaluation points of $F$ defined as
\[
    Z \coloneqq \sum_{x\in\field^n} F(x).
\]
If the new system successfully isolates a single solution, the resulting sum $\widetilde{Z}$ is $0$-$1$-valued and equals $1$ if and only if the original system has a solution.
For more details, refer to Section~\ref{sec:equations-to-sum}.
Note that without the isolation technique, $Z$ may be zero over $\field$ even if the original system has a solution.

\paragraph*{Partial Sum Polynomials.}

Another insight of Björklund et al.~\cite{DBLP:conf/icalp/BjorklundK019} is that any given instance of \SUMqd{q}{d} can be reduced to many smaller instances of \SUMqd{q}{d}.
Each of these instances is obtained by fixing a subset of variables to a particular value.
More precisely, let $\beta = \lceil \kappa n\rceil$ for an appropriately chosen $\kappa \in (0,1)$ and partition the variables into two disjoint subsets of size $n-\beta$ and $\beta$, respectively.
This allows us to express the full sum as $Z = \sum_{y\in\field^{n-\beta}} \sum_{z\in\field^{\beta}} F(y,z)$.
For all fixed~$y \in \field^{n-\beta}$, define $Z_{\beta}(y)$ as the \emph{partial sum}
\[
    Z_{\beta}(y) \coloneqq \sum_{z\in\field^{\beta}} F(y,z)\,,
\]
which can be interpreted as a polynomial~$Z_\beta$ over the variables $y \in \field^{n-\beta}$ (see Section~\ref{sec:sum-to-partial-sum}) and can be computed recursively.

How does this setup help improve the running time?
Clearly, computing~$Z$ naïvely by evaluating $Z_\beta$ at all~$q^{\beta}$ evaluation points and then applying the identity~$Z=\sum_{y\in\field^{n-\beta}} Z_\beta(y)$ in time $\OOstar(q^{n-\beta})$ is no better than an $\OOstar(q^n)$-time exhaustive search.
However, the crucial insight is that the Razborov--Smolensky technique can be applied to obtain probabilistic polynomials $\widetilde{P}_i$, leading to probabilistic indicator and partial sum polynomials $\widetilde{F}$ and $\widetilde{Z}_\beta$ of lower degree.

\paragraph*{Error Correction via Plurality Votes.}

Since we use the technique of Razborov--Smolensky, the obtained probabilistic polynomials~$\widetilde P_i\in\field[X_1,\dots,X_n]$ are \emph{with high probability} correct for each assignment.
Let~$\Delta \in \N$ be the degree of the corresponding indicator polynomial~$\widetilde{F}$.
By definition, the corresponding partial sum polynomials $\widetilde Z_{\beta}$ have degree at most $\Delta$.
Consequently, we can uniquely determine $\widetilde Z_{\beta}$ by its evaluations for all ${\Delta}$-bounded points in $\field^{n-\beta}$, using known techniques for interpolation (see the paragraph on fast trimmed interpolation and multipoint evaluation below, as well as Section~\ref{sec:multipoint}).

However, assuming we obtain all necessary evaluations of $\widetilde{Z}_\beta$, it is not possible to sum over them directly to get the full sum with high probability:
While Razborov--Smolensky guarantees that the polynomial~$Z_\beta$ is point-wise correct with high probability, it in general has lower degree than the original partial sum polynomial, so we cannot expect it to be correct for all assignments with non-zero probability.
Therefore, we boost the probability of success by independently repeating $t\approx n$~times the computation of the partial sum polynomial $\widetilde Z_\beta$.
Separately for all evaluation points of $Z_\beta$, we compute \emph{plurality votes} over the $t$ independently obtained evaluations of $\widetilde{Z}_\beta$, which is very likely to correct all potential errors, and we have $\sum_{y\in\field^{n-\beta}} \widetilde Z_\beta(y) = Z$ with high probability.
Refer to Section~\ref{sec:error-probability} for more details on the error probability of our algorithm.

\paragraph*{Reduced Number of Evaluations via Symbolic Interpolation.}

Another observation of Björklund et al.~\cite{DBLP:conf/icalp/BjorklundK019} is that for every fixed $y \in \field^{n-\beta}$, the computation of $\widetilde Z_{\beta}(y)$ can be further simplified to
\[
    \widetilde Z_{\beta}(y) = \sum_{z\in\field^{\beta}} \left(\widetilde F_1(y) \cdot \prod_{i=1}^{\beta} z_i^{q-1}\right) = (q-1)^{\beta} \cdot \widetilde F_1(y)
\]
for an appropriately chosen polynomial $\widetilde F_1$.
Refer to Lemma~\ref{lem:modqtrick} for the definition of $\widetilde F_1$.
This can be seen as a \emph{symbolic interpolation}, since the polynomial $\widetilde F_1$ can be viewed as the \emph{symbolic coefficient} of the monomial~$\prod_{i=1}^{\beta} z_i^{q-1}$ in $\widetilde{F}$.
In other words, this reduces the number of points on which the partial sum polynomial~$\widetilde Z_{\beta}$ is evaluated on in Lemma~\ref{lem:Zbeta-degree} to all $({\Delta} - \beta)$-bounded points in $\field^{\beta}$ instead of all  $\Delta$-bounded points in $\field^{\beta}$.

\paragraph*{Efficient Computation of Partial Sum Polynomials.}

The key insight of Dinur~\cite{dinur2021improved} is that all of the smaller instances of \SUMqd{q}{d} are actually \emph{related} and solving them independently like Björklund et al.~\cite{DBLP:conf/icalp/BjorklundK019} is suboptimal.
Let $\beta' = \beta - \lceil \lambda n \rceil$ for an appropriately chosen $\lambda \in (0,1)$, then the variable subset of size~$\beta$ is further divided into two disjoint subsets of size~$\beta-\beta'$ and $\beta'$.
How does this additional variable partition relate to Dinur's insight?
Let $\Delta' \in \N$ and define the degree-$\Delta'$ partial sum polynomial as
\[
    \widetilde Z_{\beta'}(y,u) \coloneqq \sum_{v \in \field^{\beta'}} \widetilde F\big((y,u), v \big)
\]
for all fixed $(y,u) \in \field^{n-\beta'}$.
This nicely illustrates the fact that all smaller instances of \SUMqd{q}{d} (each instance corresponds to a fixed assignment~$(y,u) \in \field^{n-\beta'}$) are related, because all instances are asking for sums over all evaluations on assignments $v \in \field^{\beta'}$ of the \emph{same} polynomial.
Instead of considering all assignments of $v \in \field^{\beta'}$ independently for every single instance, it is more efficient to consider them \emph{exactly once} for all $(y,u) \in \field^{n-\beta'}$ combined (for more details, see Section~\ref{sec:PartialSums}).
This leads to the idea to have our recursive algorithm compute the polynomial $\widetilde{Z}_\beta$ (followed by $\widetilde{Z}_{\beta'}$ and so on) as a whole, whose evaluations provide us with \emph{all} the required partial sums.

Recall that $\widetilde  Z_{\beta'}$ of degree $\Delta'$ is uniquely determined by its evaluations for all ${\Delta'}$-bounded points in $\field^{\beta'}$.
As a consequence, the evaluation of $\widetilde Z_{\beta'}$ for all $\Delta'$-bounded points~$(y,u) \in \field^{n-\beta'}$ takes a \emph{single recursive call} to the algorithm that we are constructing.
This stands in stark contrast to the exponentially many recursive calls necessary in the algorithm of Björklund et al.~\cite{DBLP:conf/icalp/BjorklundK019}.
For a detailed visualization of the recursion schemes of Björklund et al.~\cite{DBLP:conf/icalp/BjorklundK019} and Dinur~\cite{dinur2021improved} and by extension our algorithm refer to Figure~\ref{fig:approach-comparison}.

\paragraph*{Fast Trimmed Interpolation and Multipoint Evaluation.}

As further explained in Section~\ref{sec:multipoint}, there are a variety of interpolation and multipoint evaluation algorithms that allow us to switch between the evaluation and the coefficient representation of a polynomial.
The recursion scheme introduced by Dinur~\cite{dinur2021improved} repeatedly interpolates and evaluates the $\Delta'$-bounded polynomial~$\widetilde Z_{\beta'}$ on \emph{exactly} the set of degree-bounded points that is required to uniquely determine $\widetilde  Z_{\beta'}$.
Computing the full evaluation representation on every recursion level like Björklund et al.~\cite{DBLP:conf/icalp/BjorklundK019} is suboptimal.
Instead, we use the \emph{fast trimmed interpolation and multipoint evaluation} algorithms by Van der Hoeven and Schost~\cite{DBLP:journals/aaecc/HoevenS13} with a running time that is linear in the number of provided evaluations or coefficients of a polynomial, respectively (refer to Lemma~\ref{lem:interpolation}).

More precisely, we compute the coefficient representation of $\widetilde  Z_{\beta'}$ from its evaluations for all $\Delta'$-bounded points~$(y,u) \in \field^{n-\beta'}$ in the recursive call.
With its coefficient representation, we can further evaluate $\widetilde Z_{\beta'}$ on all points $\TrimmedFp{n-\beta}{\Delta} \times \field^{\beta - \beta'}$, where $\TrimmedFp{n-\beta}{\Delta}$ is the set of $\Delta$-bounded points in $\field^{n-\beta}$.
After using plurality votes to correct possible mistakes, we compute the evaluations of the polynomial
\[
  \widetilde  Z_{\beta}(y) \coloneqq \sum_{u \in \field^{\beta-\beta'}} \widetilde Z_{\beta'}(y,u)
\]
on all $\Delta$-bounded points~$y \in \field^{n-\beta}$ and finally interpolate to obtain the polynomial~$\widetilde{Z}_\beta$.
This polynomial is with high probability identical to the partial sum polynomial $Z_{\beta}$.


\pgfdeclarepatternformonly{sparse vertical lines}{\pgfqpoint{-0pt}{-0pt}}{\pgfqpoint{6pt}{6pt}}{\pgfqpoint{6pt}{6pt}}{
        \pgfsetlinewidth{0.75pt}
        \pgfpathmoveto{\pgfqpoint{1.5pt}{0pt}}
        \pgfpathlineto{\pgfqpoint{1.5pt}{6pt}}
        \pgfusepath{stroke}}

\pgfdeclarepatternformonly{dense vertical lines}{\pgfqpoint{-0pt}{-0pt}}{\pgfqpoint{3pt}{3pt}}{\pgfqpoint{3pt}{3pt}}{
        \pgfsetlinewidth{0.75pt}
        \pgfpathmoveto{\pgfqpoint{1.5pt}{0pt}}
        \pgfpathlineto{\pgfqpoint{1.5pt}{3pt}}
        \pgfusepath{stroke}}

\begin{figure}
    \begin{tikzpicture}
        [
            auto,
            level 1/.style={sibling distance=15mm},
            level 2/.style={sibling distance=7mm},
            nodey/.style = {
                    font = \tiny,
                    align = center,
                    inner sep = 0.1em,
                    outer sep = 0em
                },
            nodex/.style = {
                    align = center,
                    inner sep = .2em,
                }
        ]
        \matrix [column sep=2em, row sep=1.5em]
        {

            \node []
            {
                \begin{minipage}[t][]{.4\textwidth}
                    \centering
                    {\sffamily\bfseries Björklund, Kaski, and Williams \cite{DBLP:conf/icalp/BjorklundK019}}
                \end{minipage}
            };
             & \node []
               {
                   \begin{minipage}[t][]{.4\textwidth}
                    \centering
                    {\sffamily\bfseries Dinur~\cite{dinur2021improved} \& this paper}
                \end{minipage}
               };                              \\

            \node [] {}
            child {node (b1) [rectangle,draw, fill=lightblue, minimum width=3cm, minimum height=0.58cm, yshift=1.25cm] {$\scriptstyle B_1$} edge from parent[draw=none] }
            child {node [inner sep = 0em, outer sep = 0em] {} edge from parent[draw=none]}
            child {node [rectangle,draw, minimum width=2cm, minimum height=0.58cm, yshift=1.25cm] {$\scriptstyle A_1$} edge from parent[draw=none]
                    child {node [rectangle,draw, fill=lightblue, minimum width=1.25cm, minimum height=0.58cm] {$\scriptstyle B_2$}}
                    child {node [inner sep = 0em, outer sep = 0em] {} edge from parent[draw=none]}
                    child {node [rectangle,draw, minimum width=0.75cm, minimum height=0.58cm] {$\scriptstyle A_2$}
                            child {node [rectangle,draw, fill=lightblue, minimum width=0.5cm, minimum height=0.58cm] {}}
                            child {node  [rectangle,draw, minimum width=0.25cm, minimum height=0.58cm] {}
                                    child {node {} edge from parent[dotted, line width= 0.75, shorten >= 0.7cm]}
                                    child {node {} edge from parent[dotted, line width= 0.75, shorten >= 0.7cm]}
                                }
                        }
                };

            \begin{scope}
                \draw[-stealth] (-1.5,-0.55) |-  (0.17,-1.55);
                \draw[-stealth] (-1.65,-0.55) node[below, xshift=-0.6cm, yshift=-0.5cm] {$\scriptstyle\cdots$} |- node[below, xshift=.7cm, yshift=-0.12cm, rotate=90, font=\fontsize{3}{3.5}\selectfont] {} (0.17,-1.7);
                \draw[-stealth] (-2.9, -0.55) |- node[below, xshift=1.5cm, yshift=-.1cm] {$\binom{B_1}{\downarrow \Delta_1}_q$ {\scriptsize calls}} (0.17,-1.95);

                \draw[-stealth] (0.7,-2.05) |- (1.6,-3.05);
                \draw[-stealth] (0.55,-2.05) node[below, xshift=-0.12cm, yshift=-0.5cm] {} |- node[below, xshift=0.3cm, yshift=-0.12cm, rotate=90, font=\fontsize{3}{3.5}\selectfont] {} (1.6,-3.2);
                \draw[-stealth] (0.3,-2.05) |- node[below, xshift=0.6cm, yshift=-.1cm] (b2){$\binom{B_2}{\downarrow \Delta_2}_q$ {\scriptsize calls}} (1.6,-3.45);
                \pgfresetboundingbox
                \path[use as bounding box] (-3,0.05) rectangle (2.75,-4.25);
            \end{scope}

             &
            \node [] {}
            child {node (d1)  [rectangle, draw, fill=lightblue, minimum width=2.8cm, minimum height=0.58cm, yshift=1.25cm] {$\scriptstyle (1 - \kappa) n$} edge from parent[draw=none]
                    child {node [pattern= dense vertical lines, pattern color=lightblue, rectangle, draw, minimum width=2.8cm, minimum height=0.58cm] {} edge from parent[draw=none]
                            child {node [pattern= sparse vertical lines, pattern color=lightblue, rectangle, draw, minimum width=2.8cm, minimum height=0.58cm] {} edge from parent[draw=none]}
                        }
                }
            child {node [inner sep = 0em, outer sep = 0em] {} edge from parent[draw=none]}
            child {node [rectangle,draw, minimum width=2.5cm, minimum height=0.58cm, yshift=1.25cm] {$\scriptstyle \kappa n$} edge from parent[draw=none]
                    child {node [pattern= dense vertical lines, pattern color=lightblue, rectangle, draw, minimum width=1.4cm, minimum height=0.58cm] {}
                            child {node [pattern=sparse vertical lines, pattern color=lightblue, rectangle, draw, minimum width=1.4cm, minimum height=0.58cm] {} edge from parent[draw=none]}
                        }
                    child {node [inner sep = 0em, outer sep = 0em] {} edge from parent[draw=none]}
                    child {node [rectangle, draw, minimum width=0.9cm, minimum height=0.58cm] {$\scriptstyle (\kappa - \lambda)n$}
                            child {node [pattern= sparse vertical lines, pattern color=lightblue, rectangle, draw, minimum width=0.55cm, minimum height=0.58cm] {}}
                            child {node [rectangle,draw, minimum width=0.35cm, minimum height=0.58cm] {}
                                    child {node {} edge from parent[dotted, line width= 0.75, shorten >= 0.7cm]}
                                    child {node {} edge from parent[dotted, line width= 0.75, shorten >= 0.7cm]}
                                }
                        }
                };

            \begin{scope}
                \draw (-1.5,-0.6) -- (-1.5,-0.6) node[below] {$\scriptstyle (1-\kappa) n$ {\scriptsize evals}};
                \draw (-0.7,-2) -- (-0.7,-2) node[below] {$\scriptstyle \binom{(1 - (\kappa - \lambda))n}{\downarrow \Delta_1}_q$ {\scriptsize evals}};
                \draw (-0.2,-3.5) -- (-0.2,-3.5) node[below] {$\scriptstyle \binom{(1 - (\kappa - 2\lambda))n}{\downarrow \Delta_2}_q$ {\scriptsize evals}};
                \pgfresetboundingbox
                \path[use as bounding box] (-3,0.05) rectangle (2.75,-4.25);
            \end{scope}
            \\

            \node [] {
                \begin{minipage}[t][]{.45\textwidth}
                    For each recursion level $i$:
                    \begin{itemize}[noitemsep]
                        \item The number of variables $B_i$ that have a fixed value \emph{decreases}
                        \item The number of variables $A_i$ that are free \emph{decreases}
                        \item The \emph{full evaluation representation }of a degree-$\Delta_i$ $|A_i|$-variate polynomial $P$
                        can be obtained from all points in $\TrimmedFp{A_i}{\Delta_i}$
                        \item The degree of $P$ \emph{decreases} with each recursion level, that is, $\Delta_{i}> \Delta_{i+1}$
                        \item $t \cdot\binom{B_i}{\downarrow \Delta_i}_q$ recursive calls to level $i\!+\!1$
                    \end{itemize}
                \end{minipage}
            };
             & \node []
               {
                   \begin{minipage}[t][]{.45\textwidth}
                    For each recursion level $i$:
                    \begin{itemize}[noitemsep]
                        \item The number of variables $n' =(1-(\kappa-i\lambda))n$, for which bounded evaluations are computed, \emph{increases}
                        \item The number of variables $n_i = (\kappa -i\lambda) n$ that are free \emph{decreases}
                        \item The $\Delta_{i}$\emph{-bounded representation} of a degree-$\Delta_i$ $n'$-variate polynomial $P$
                        can be obtained from all points in $\TrimmedFp{n'}{\Delta_{i+1}}$
                        \item The degree of $P$ \emph{decreases} with each recursion level, that is, $\Delta_{i}> \Delta_{i+1}$
                        \item $t$ recursive calls to level $i\!+\!1$
                    \end{itemize}
                \end{minipage}
               };                             \\

            \node [] {\PartialSum{\ast}}
            child {node [rectangle,draw] {$j=1$}
                    child {node [nodey] {} edge from parent[dotted, line width= 0.75]}
                    child {node [nodey] {} edge from parent[dotted, line width= 0.75]}
                    child {node [nodey, transform canvas={yshift=2.5mm}] {$\cdots$} edge from parent[draw=none]}
                    child {node [nodey] {} edge from parent[dotted, line width= 0.75]}
                    child {node [nodey] {} edge from parent[dotted, line width= 0.75]}
                }
            child {node {$\cdots$} edge from parent[draw=none]}
            child {node [rectangle,draw] {$j = t$}
                    child {node (11) [nodey] {}
                            child {node {} edge from parent[dotted, line width= 0.75, shorten >= 1.1cm]}
                        }
                    child {node [nodex] {$B_1$}
                            child {node (21) [nodey] {}
                                    child {node {} edge from parent[dotted, line width= 0.75, shorten >= 1.1cm]}
                                }
                            child {node [nodex] {$B_2$}
                                    child {node (31) [nodey] {}
                                            child {node {} edge from parent[dotted, line width= 0.75, shorten >= 1.1cm]}
                                        }
                                    child {node [nodey] {}
                                            child {node {} edge from parent[dotted, line width= 0.75, shorten >= 1.1cm]}
                                        }
                                    child {node [nodey, transform canvas={yshift=2.5mm}] {$\cdots$} edge from parent[draw=none]}
                                    child {node [nodey] {}
                                            child {node {} edge from parent[dotted, line width= 0.75, shorten >= 1.1cm]}
                                        }
                                    child {node (32) [nodey] {}
                                            child {node {} edge from parent[dotted, line width= 0.75, shorten >= 1.1cm]}
                                        }
                                }
                            child {node [nodey, transform canvas={yshift=2.5mm}] {$\cdots$} edge from parent[draw=none]
                                }
                            child {node [nodey] {}
                                    child {node {} edge from parent[dotted, line width= 0.75, shorten >= 1.1cm]}
                                }
                            child {node (22) [nodey] {}
                                    child {node {} edge from parent[dotted, line width= 0.75, shorten >= 1.1cm]}
                                }
                        }
                    child {node [nodey, transform canvas={yshift=2.5mm}] {$\cdots$} edge from parent[draw=none]}
                    child {node [nodey] {}
                            child {node {} edge from parent[dotted, line width= 0.75, shorten >= 1.1cm]}
                        }
                    child {node (12) [nodey] {}
                            child {node {} edge from parent[dotted, line width= 0.75, shorten >= 1.1cm]}
                        }
                };

            \begin{scope}
                \draw[bend right, transform canvas={yshift=10mm}, shorten <=0.5cm, shorten >=0.5cm] (11) edge node[left, pos=0.15] {$\binom{B_1}{\downarrow \Delta_1}_q$} (12);
                \draw[bend right, transform canvas={yshift=10mm}, shorten <=0.5cm, shorten >=0.5cm] (21) edge node[left, pos=0.15] {$t$} (22);
                \draw[bend right, transform canvas={yshift=10mm}, shorten <=0.5cm, shorten >=0.5cm] (31) edge node[left, pos=0.15] {$\binom{B_2}{\downarrow \Delta_2}_q$} (32);
                \pgfresetboundingbox
                \path[use as bounding box] (-2.9,0.2) rectangle (2.95,-6.25);
            \end{scope}

             & \node [] {\PartialSum{\ast}}
            child {node [rectangle,draw] {$j=1$}
                    child {node (1) [nodey] {} edge from parent[draw=none]}
                    child {node [nodex] {$(\kappa - \lambda)n$}
                            child {node [nodex] {$(\kappa - 2\lambda)n$}
                                    child {node {} edge from parent[dotted, line width= 0.75, shorten >= 0.7cm]}
                                }
                        }
                    child {node (2) [nodey] {} edge from parent[draw=none]}
                }
            child {node {$\cdots$} edge from parent[draw=none]}
            child {node [rectangle,draw] {$j = t$}
                    child {node (3) [nodey] {} edge from parent[draw=none]}
                    child {node [nodex] {$(\kappa - \lambda)n$}
                            child {node [nodex] {$(\kappa - 2\lambda)n$}
                                    child {node {} edge from parent[dotted, line width= 0.75, shorten >= 0.7cm]}
                                }
                        }
                    child {node (4) [nodey] {} edge from parent[draw=none]}
                };

            \begin{scope}
                \draw[bend right, transform canvas={yshift=10mm}, shorten <=0.4cm, shorten >=0.4cm] (1) edge node[left, pos=0.25] {$\scriptstyle t$} (2);
                \draw[bend right, transform canvas={yshift=-5mm}, shorten <=0.4cm, shorten >=0.4cm] (1) edge node[left, pos=0.25] {$\scriptstyle t$} (2);

                \draw[bend right, transform canvas={yshift=10mm}, shorten <=0.4cm, shorten >=0.4cm] (3) edge node[right, pos=0.75] {$\scriptstyle t$}  (4);
                \draw[bend right, transform canvas={yshift=-5mm}, shorten <=0.4cm, shorten >=0.4cm] (3) edge node[right, pos=0.75] {$\scriptstyle t$}  (4);
                \pgfresetboundingbox
                \path[use as bounding box] (-2.1,0.2) rectangle (2.1,-5.2);
            \end{scope}                                \\
        };
    \end{tikzpicture}
    \caption{\label{fig:approach-comparison}%
        We compare the algorithm of Björklund, Kaski, and Williams~\cite{DBLP:conf/icalp/BjorklundK019} with the one of Dinur~\cite{dinur2021improved} (which is the basis for our algorithm).
        The main differences lie in how the variables are partitioned and how the partial sum polynomial is computed.
        In the schematic of the variable partitioning, each box with a set of variables represents the monomials of the indicator polynomial in these variables and the colored entries illustrate the computed evaluations at the current recursion level.
        Note that both recursion schemes have to be repeated $t$ times to boost the outcome probability of Razborov--Smolensky.}
\end{figure}

\section{Preliminaries}

In this section, we introduce the required definitions and preliminary results that allow us to construct our main algorithm and prove its error probability and running time.
A large portion of this section is spent on giving an approximation of the number of monomials of any polynomial $P \in \field[X_1, \dots, X_n]$ of degree at most $d$, see~Section~\ref{sec:approx-number-monomials}.

\subsection{Chernoff Bounds.}
We use the following standard Chernoff bound.

\begin{lemma}[e.g., Mitzenmacher and Upfal {\cite[Theorem~4.5]{mitzenmacher2017probability}}]\label{lem:chernoff}
    Let $X_1,\dots,X_n$ be independent random variables on~$\set{0,1}$ and let $X = X_1+\dots+X_n$.
    For all $\delta$ with $0<\delta<1$, we have
    \[
        \Pr\Big(X \le (1-\delta) \, \E(X)\Big) \le \exp \Big(- \frac{\delta^2 \, \E(X)}{2}\Big).
    \]
\end{lemma}

\subsection{Polynomial Rings.}
In this paper, we consider polynomials in the polynomial ring $\field[X_1,\dots,X_n]$ for some fixed prime power $q$.
If $q$ is prime, Fermat's little theorem states $a^{q}=a$ for every~$a\in\field$; equivalently, $a^{q-1}=1$ holds for all $a\in\field\setminus\set{0}$.
This is known to generalize to prime powers~$q$ by applying Lagrange's Theorem (see Lang~\cite[Chapter~I, Proposition~2.2]{DBLP:books/daglib/0070572}) to the subgroup of elements generated by $a$.
Thus, we can restrict our attention to polynomials that have degree at most~$q-1$ in each variable.
We use the following corollary to Fermat's little theorem.
\begin{lemma}\label{lem:fermat-coro}
    For all prime powers~$q$ and all $k\in\set{0,\dots,q-1}$, we have
    \begin{equation}
        \sum_{x\in\field} x^k =
        \begin{cases}
            q-1 & \text{if $k=q-1$, and} \\
            0   & \text{otherwise.}
        \end{cases}
    \end{equation}
\end{lemma}
\begin{proof}
    For $k=0$, we have $x^k=1$ for all $x\in\field$, and thus $\sum_x x^0=q=0$ holds as claimed.
    For $k>0$, we have $0^k=0$, so only the $q-1$ summands with~$x\ne 0$ contribute to the sum.
    In particular, for $k=q-1$, Fermat's little theorem implies $x^{q-1}=1$ for all $x\in\field$ with $x\ne 0$, and thus $\sum_{x\in\field} x^{q-1}=q-1$.

    Now suppose $1\le k\le q-2$.
    Note that the multiplicative group $\field\setminus\set{0}$ is cyclic (see Lang~\cite[Chapter~IV, Corollary~1.10]{DBLP:books/daglib/0070572}), that is, it is generated by a single element~$g\in\field\setminus\set{0}$.
    Consequently, we have
    \[
        \sum_{x\in\field} x^{k}=
        \sum_{x\in\field\setminus\set{0}} x^{k}=
        \sum_{i\in\set{0,\dots,q-2}} \paren{g^i}^{k}=
        \sum_{i\in\set{0,\dots,q-2}} \paren{g^{k}}^{i}=
        \frac{\paren{g^{k}}^{q-1}-1}{g^{k}-1}=
        \frac{1-1}{g^k-1}
        =0\,.
    \]
    Here, the penultimate equality again follows from Fermat's little theorem.
    Also note $g^k \neq 1$ for $k < q-1$, because~$g$ generates $\field$.
    This concludes the proof of the lemma.
\end{proof}

For a vector~$M\in\set{0,\dots,q-1}^n$, we write $X^M$ for the monomial $\prod_{i=1}^n X_i^{M_i}$, and analogously for an assignment $x\in\field^n$ to the vector of variables $X$.
Each $n$-variate polynomial $P\in\field[X_1,\dots,X_n]$ has the form
\[
    P(X_1,\dots,X_n)=\sum_{M\in\set{0,\dots,q-1}^n} c_M X^M
\]
for some values $c_M \in \field$.
The values~$c_M$ are called \emph{coefficients} of~$P$.

In our algorithm, we will compute the sum $\sum_{a\in\field^n}P(a)$ over all evaluation points~$a$ of a polynomial~$P$.
The following lemma implies that this is equivalent to determining the coefficient of the monomial $\prod_{i=1}^n x_i^{q-1}$ in~$P$. That is, only this one coefficient contributes to the sum $\sum_{a\in\field^n}P(a)$ and all other coefficients cancel out.
\begin{lemma}\label{lem:modp-sum}
    Let $q$ be a prime power and $M \in \set{0,\dots,q-1}^n$ for $n\in\N$. Then we have
    \begin{align*}
        \sum_{x \in \field^n} \prod_{i=1}^n x_i^{M_i} =
        \begin{cases}
            (q-1)^{n}, & \text{if } M_i = q-1 \text{ for all } i \in [n] \\
            0,         & \text{otherwise.}
        \end{cases}
    \end{align*}
\end{lemma}
\begin{proof}
    We prove the statement by induction. For $n=1$, the statement is that of Lemma~\ref{lem:fermat-coro}.
    Assume the statement holds for some $n\in\N$.
    Let $M \in \set{0,\dots,q-1}^{n+1}$ and observe that distributivity implies the following:
    \begin{align*}
        \sum_{x\in\field^{n+1}} \prod_{i=1}^{n+1} x_i^{M_i} = \left( \sum_{y\in\field} y^{M_{n+1}} \right) \cdot \left( \sum_{x\in\field^{n}} \prod_{i=1}^{n} x_i^{M_i} \right).
    \end{align*}
    The statement follows by applying Lemma~\ref{lem:fermat-coro} to the first and the induction hypothesis to the second factor of the product of sums above.
\end{proof}

With the insights of Lemma~\ref{lem:fermat-coro} and Lemma~\ref{lem:modp-sum}, we can show that for every fixed $a \in \field^{n_1}$ the sum $\sum_{b \in \field^{n_2}} P(a,b)$ over all evaluation points~$b$ can be computed by $(q-1)^{n_2} \cdot P_1(X)$, where $P_1(X)$ is obtained from $P$ by setting certain coefficients to $0$.
Since the polynomial $P_1$ can be seen as the symbolic coefficient of the monomial $\prod_{i=1}^{n_2} b_i^{q-1}$ in $P$, this can be seen as a symbolic interpolation.

\begin{lemma}[Symbolic Interpolation]\label{lem:modqtrick}
    Let $q$ be a prime power, $n_1, n_2 \in \N$, $X = (X_1, \dots, X_{n_1})$, $Y = (Y_1, \dots, Y_{n_2})$, $P \in \field[X,Y]$ and let $c_M$ for $M \in \{0, \dots, q-1\}^{n_1 + n_2}$ be the coefficients of $P$.
    Define
    \[P_1(X) \coloneqq \sum_{M \in \{0, \dots, q-1\}^{n_1} \times \{q-1\}^{n_2}} c_M \cdot \prod_{i=1}^{n_1} X_i^{M_i}\,.\]
    Then we have
    \[P_1(X) = (q-1)^{n_2}\cdot\sum_{y \in \field^{n_2}} P(X, y)\,.\]
\end{lemma}
\begin{proof}
    Define the set
    \(\mathrm{Mon}_0 \coloneqq \{M \in \{0, \dots, q-1\}^{n_1 + n_2} \mid \text{there is } i \in [n_2] \text{ with } M_{n_1 + i} \neq q-1\}\)
    and the polynomial $P_0(X,Y) \coloneqq \sum_{M \in \mathrm{Mon}_0} c_M (X,Y)^M$.
    Now for any fixed $x \in \field^{n_1}$, the sum of all evaluations of $P$ under the partial assignment $x$ can be written as follows:
    \begin{align*}
        \sum_{y \in \field^{n_2}} P(x, y)
         & = \sum_{y \in \field^{n_2}} P_0(x, y) + \sum_{y \in \field^{n_2}} \left( \prod_{i=1}^{n_2} y_i^{q-1} \cdot P_1(x) \right) \\
         & = \sum_{y \in \field^{n_2}} P_0(x, y) + P_1(x) \cdot \sum_{y \in \field^{n_2}} \prod_{i=1}^{n_2} y_i^{q-1}.
    \end{align*}

    We now show that the first sum is actually equal to $0$. To this end, we write

    \begin{align*}
        \sum_{y \in \field^{n_2}} P_0(x, y)
         & = \sum_{y \in \field^{n_2}} \left( \sum_{M \in \mathrm{Mon}_0} c_M \cdot \prod_{i=1}^{n_1} x_i^{M_i} \prod_{i=1}^{n_2} y^{M_{n_1 + i}} \right)        \\
         & = \sum_{M \in \mathrm{Mon}_0} c_M \cdot \prod_{i=1}^{n_1} x_i^{M_i} \cdot \left( \sum_{y \in \field^{n_2}} \prod_{i=1}^{n_2} y_i^{M_{n_1+i}} \right).
    \end{align*}
    Using the fact that every non-zero monomial of $P_0$ has at least one variable in $Y$ whose exponent is not $q-1$, the innermost sum is $0$ by Lemma~\ref{lem:modp-sum}.
    Hence, the whole sum is equal to~$0$, and we obtain
    \[\sum_{y \in \field^{n_2}} P(x, y) = P_1(x) \cdot \sum_{y \in \field^{n_2}} \prod_{i=1}^{n_2} y_i^{q-1} = P_1(x) \cdot (q-1)^{n_2},\]
    where the last equation is obtained by applying Lemma~\ref{lem:modp-sum} again.
    The desired result follows from the fact that $q-1$ is idempotent in $\field$.
\end{proof}

\subsection{Approximation of the Number of Monomials.}
\label{sec:approx-number-monomials}

For the running time analysis of our algorithm, we bound the number of different monomials in $\field[X_1,\dots,X_n]$ with degree exactly $\Delta$.
This number is also known as the \emph{extended binomial coefficient} $\binom{n}{\Delta}_q$.
Eger~\cite[Equation~(2)]{DBLP:journals/tamm/Eger14} defined this number formally using multinomial coefficients.
\begin{Definition}[Eger \cite{DBLP:journals/tamm/Eger14}]
    For $k_1, \dots, k_q\in\N$,
    the multinomial coefficient $\binom{n}{k_1, \dots , k_q}$ satisfies
    \begin{align}
        \label{eq:def-multinomial}
        \binom{n}{k_1,\dots, k_q} 
        &= \frac{n!}{k_1! \cdot \dots \cdot k_q!}
        \,.
        \intertext{We define the extended binomial coefficient $\binom{n}{\Delta}_q$ via}
        \binom{n}{\Delta}_q
         &\label{eq:def-extended-binomial} =
        \sum_{k_1,\dots,k_q}
        \binom{n}{k_1,\dots,k_q}
        \,,
        \intertext{where the sum is taken over all $k_1, \dots, k_q\in\N$ that satisfy the constraints $\sum_{i=1}^q k_i=n$ and $\sum_{i=1}^q k_i\cdot (i-1)=\Delta$.
        For notational convenience, we also define $\binom{n}{\downarrow\Delta}_q$ via}
        \label{eq:def-extended-binomial-downarrow}
        \binom{n}{\downarrow\Delta}_q 
        &= \sum_{k=0}^\Delta \binom{n}{k}_q\,.
    \end{align}
\end{Definition}
Each multinomial coefficient $\binom{n}{k_1,\dots,k_q}$ on the right side of \eqref{eq:def-extended-binomial} corresponds to the number of monomials over $\field[X_1,\dots,X_n]$ that, for all $i\in\set{1,\dots,q}$, have $k_i$ variables of individual degree $i-1$.
Thus,
$\binom{n}{\Delta}_q$ is the number of monomials with total degree exactly $\Delta$ and
$\binom{n}{\downarrow\Delta}_q$ is the number of monomials with total degree at most~$\Delta$.
For $q=2$, we observe $\binom{n}{\Delta}_2=\binom{n}{\Delta}$, and thus we obtain the classical binomial coefficient.

Let $\TrimmedFp{n}{\Delta}$ be the set of vectors $(x_1, \dots, x_n) \in \field^n$ that satisfy $\sum_{i=1}^n x_i \leq \Delta$ over $\N$. Then each such vector corresponds to a monomial from $\field[X_1,\dots,X_n]$ of degree at most~$\Delta$, and we have $\binom{n}{\downarrow\Delta}_q=\abs{\TrimmedFp{n}{\Delta}}$.
Moreover, it is easy to see that $\binom{n}{\Delta}_q$ is increasing in~$n$ and increasing in~$\Delta$ for $\Delta\in\set{0,\dots,\floor{n(q-1)/2}}$, and that the symmetry $\binom{n}{\Delta}_q=\binom{n}{n-\Delta}_q$ holds.
Furthermore, we have $\binom{n}{\Delta}_q\leq\binom{n}{\downarrow n}_q=q^n$.

It is well-known that the classical binomial coefficient can be bounded using the binary entropy function (e.g., see \cite[Theorem~1]{ArriataGordon}):
\begin{align}
    \binom{n}{\Delta} \leq 2^{n\cdot H(\Delta/n)}
    \text{ where }
    H(p) = -p\log_2(p) - (1-p)\log_2(1-p)\,.
\end{align}
Moreover, this bound is tight up to a factor of $\Theta(\sqrt{n})$ (e.g., see \cite[Theorem~2]{ArriataGordon}).

The extended binomial coefficient has an analogous bound, but the bound does not generally have a closed-form expression. The following lemma appears to be folklore (see \cite[Proposition~4.12]{Blasiak2017On}).
We provide a proof for completeness; the proof can be seen as a generalization of a known proof for the binomial coefficient (see \cite[Proof~1]{ArriataGordon}).
\begin{lemma}\label{lem:ext-binom-bound}
    Let $q\geq 2$ be an integer and let $\alpha \in (0, \frac{1}{2})$.
    For all $n\in\N$, we have
    \begin{align*}
        \binom{n}{\downarrow\alpha(q-1)n}_q
        &\le
        \left(\inf_{0 < x < 1} \frac{x^0 + \dots + x^{q-1}}{x^{\alpha(q-1)}} \right)^n
        =q^{H(q,\alpha)\cdot n}\,,
        \,\text{where}\\
        H(q,\alpha)
        &\coloneqq
        \inf_{\theta < 0} \left( -\alpha\theta + \log_q  \frac{1 - q^{\theta q/(q-1)}}{1- q^{\theta/(q-1)}}\right)\,.
    \end{align*}
\end{lemma}
We remark that for constant~$q$ and~$\alpha$, as $n$ tends to infinity, the bound in Lemma~\ref{lem:ext-binom-bound} is tight up to factors subexponential in~$n$, which follows from Cramér's theorem~\cite[§2.4]{RAS15} in the theory of large deviations.
\begin{proof}
    Let $\Delta\coloneqq\alpha(q-1)n$.
    Since $\binom{n}{\downarrow\Delta}_q$ is the number of monomials with total degree at most~$\Delta$, it is equal to~$q^n$ times the probability that a uniformly random monomial has total degree at most~$\Delta$.
    We prove the inequality by bounding this probability.
    \begin{align*}
        \binom{n}{\downarrow\Delta}_q
        &=
        q^n\cdot
        \Pr_{d_1, \dots, d_n \in \set{0, \dots, q-1}} \left( \sum_{i=1}^n d_i \le \Delta \right) 
        =
        q^n\cdot
        \inf_{\theta<0}\;\Pr_{d_1, \dots, d_n \in \set{0, \dots, q-1}} \left( e^{\theta \sum_{i=1}^n d_i} \ge e^{\Delta \theta } \right)
        \intertext{The second equality trivially holds for all $\theta<0$; we apply Markov's inequality next.}
        &\le
        q^n\cdot
        \inf_{\theta<0}\;\frac{\E_{d_1, \dots, d_n \in \set{0, \dots, q-1}} \left( e^{\theta \sum_{i=1}^n d_i} \right)}{e^{\Delta \theta }}
        \intertext{Since all $d_i$'s are independent, the expected value is multiplicative.}
        &=
        q^n\cdot
        \inf_{\theta<0}\;\left( \frac{\E_{d\in \set{0, \dots, q-1}} \left( e^{\theta d} \right)}{e^{\alpha(q-1)\theta}} \right)^n
        \intertext{Next we apply the definition of the expected value, which cancels the $q^n$ term.}
        &= \left(\inf_{\theta < 0}  \frac{e^{0\cdot\theta} + \dots + e^{(q-1)\cdot \theta}}{e^{\alpha(q-1)\theta}}\right)^n
        \intertext{We substitute $x = e^{\theta}$; by $\theta < 0$, we have $x \in (0,1)$.}
        &= \left(\inf_{0 < x < 1} \frac{x^0 + \dots + x^{q-1}}{x^{\alpha(q-1)}} \right)^n.
    \end{align*}
    This proves the desired inequality.
    To prove that this is equal to $q^{H(q,\alpha)\cdot n}$, we use the closed form for $x^0+\dots+x^{q-1}$ and perform the substitution $x=q^{\theta/(q-1)}$ to see the following with a straightforward calculation:
    \begin{align*}
        \log_q\paren*{\inf_{0 < x < 1} \frac{x^0 + \dots + x^{q-1}}{x^{\alpha(q-1)}}}
        &=
        \inf_{0 < x < 1} \log_q\paren*{\frac{1-x^q}{(1-x) x^{\alpha(q-1)}}}\\
        &=
        \inf_{\theta<0} - \alpha\theta + \log_q\paren*{
            \frac%
            {1-q^{\theta q/(q-1)}}%
            {1-q^{\theta/(q-1)}}%
            }
        = H(q,\alpha)\,.
    \end{align*}
    This concludes the proof.
\end{proof}

We now observe that $H(q,\alpha)$ specializes to the binary entropy for $q=2$.
\begin{lemma}
    For all $\alpha\in[0,1]$, we have $H(2,\alpha)=H(\alpha)$.
\end{lemma}
\begin{proof}
    The proof follows from basic calculus. First we define the function
    \[
    f(\theta)=-\alpha\theta + \log_2\paren*{1 - 4^{\theta}} - \log_2\paren*{1- 2^{\theta}}
    \]
    and note 
    \(
    H(2,\alpha)
    =
    \inf_{\theta < 0} f(\theta)
    \).
    Using a computer algebra system such as Wolfram Alpha, we verify
    $\lim_{\theta\to-\infty} f(\theta)=\infty$ and
    $\lim_{\theta\to 0^-} f(\theta)=1$, that the derivative $f'(\theta)$ is zero if and only if $\theta=\theta^\ast\coloneqq\log_2(\alpha) - \log_2(1-\alpha)<0$ holds, and that $f(\theta^*)\leq1$ holds at this value $\theta^\ast$. Thus, this is where the infimum is attained, and we have:
    \begin{align*}
        H(2,\alpha)&=f(\theta^\ast)=-\alpha\log_2\paren*{\frac{\alpha}{1-\alpha}}
        +\log_2\paren*{\frac{1-2\alpha}{(1-\alpha)^2}} - \log_2\paren*{\frac{1-2\alpha}{1-\alpha}}\\
        &=
        -\alpha\log_2(\alpha) -(1-\alpha)\log_2(1-\alpha) = H(\alpha)\,.
    \end{align*}
    This concludes the proof.
\end{proof}
Blasiak et al.~{\cite[Proposition~4.12]{Blasiak2017On}} studied analytic properties of the function $I(q,\alpha)$ defined via $H(q,\alpha) = 1 - I(q-1,\alpha)/\ln q$.
We state their result as follows.\footnote{We remark that their statement contains a small mistake in that they write $e^{\theta-1}$ in place of $e^{\theta}-1$.}
\begin{lemma}[Blasiak et al.~{\cite[Proposition~4.12]{Blasiak2017On}}]%
    \label{lem:exponent-properties}%
    Let $\alpha \in (0, \frac{1}{2})$ be fixed.
    The function $I(q, \alpha)$ is positive, increasing in $q$ and converges to $I^\ast_\alpha\coloneqq\sup_{\theta < 0} \big(\alpha\theta - \ln(\frac{e^{\theta}-1}{\theta})\big)$ for $q \rightarrow \infty$.
\end{lemma}
We make some observations for each fixed $\alpha\in(0,\tfrac12)$. Since $I(q,\alpha)\geq0$ holds, we have ${H(q,\alpha)\leq1}$.
Moreover, the inequality in Lemma~\ref{lem:ext-binom-bound} implies $H(q,\alpha)\geq0$. Since $I(q,\alpha)$ is increasing in~$q$ and bounded by a fixed constant~$I^\ast_\alpha$, we have $\lim_{q\to\infty}H(q,\alpha) = 1$, as well as
\(
\lim_{q\to\infty} q^{H(q,\alpha)-1}=e^{-I^\ast_\alpha}
\).
Following the proof of Lemma~\ref{lem:exponent-properties}, $H(q,\alpha)$ can be seen to be increasing in $q$ for every fixed $\alpha\in(0,\tfrac12)$.

\subsection{Machine Model, Complexity, and Representations of Polynomials.}
For our algorithms, we silently use an extension of standard word-RAM machines with words of $\OO(\log N)$ bits as our machine model, where $N$ is the input length.
The time complexity of the machine is defined as usual via the number of elementary operations performed by the machine.
We remark already here that for our main algorithm, a prime power $q$ will be fixed, and consequently all arithmetic operations on $\field$ can be performed in constant time.

Since our algorithms use arrays of exponential length and since we will not want to worry about the overhead that this causes, we silently assume that the word-RAM machine in addition has access to an abstract dictionary data structure:
In particular, the algorithm can initialize a new dictionary, read a value from the dictionary, or write a value to the dictionary, and we assume each of these operations to incur unit cost. During the initialization, we can also specify a default value for the dictionary---this value will be returned if we are trying to read the value for a key that has not been written to yet. To allow for keys with $\poly(N)$ bits, we assume that the key must be written to a special query array of the word-RAM before the dictionary's read or write operation is called.

Throughout this paper, we silently represent polynomials using dictionaries.
There are two representations that we will use to store an $n$-variate degree-$d$ polynomial $P\in\field[X_1,\dots,X_n]$:
\begin{itemize}
    \item In the \emph{coefficient representation},
          we represent~$P$ as a dictionary that stores each non-zero coefficient~$c_M$ of~$P$ under the key~$M\in\set{0,\dots,q-1}^n$. Since $P$ has degree at most~$d$, each key satisfies $\sum_i M_i\le d$.
    \item In the \emph{evaluation representation},
          we represent~$P\in\field[X_1,\dots,X_n]$ as a dictionary that stores evaluations~$P(x)$ under all keys~$x\in\TrimmedFp{n}{d}$ (recall that this is the set of tuples in $\field^n$ whose entries sum up to at most $d$ over $\N$).
\end{itemize}
In the next section, we will show how to efficiently switch between these two representations.
If no representation is specified, we silently use the coefficient representation by default.

\subsection{Fast Multipoint Evaluation and Interpolation over Finite Fields.}\label{sec:multipoint}

As an important subroutine of Theorem~\ref{thm:pes-algorithm}, we use a fast algorithm for multipoint evaluation and interpolation for bounded-degree multivariate polynomials over finite fields.
There is a natural bijection between elements of $\TrimmedFp{n}{\Delta}$ and the monomials of a degree-$\Delta$ polynomial $P\in\field[X_1,\dots,X_n]$.
The polynomial~$P$ can be represented either by providing the coefficient of each monomial of degree at most~$\Delta$, or by providing the evaluations of~$P$ at all points in $\TrimmedFp{n}{\Delta}$.
Moreover, it is possible to efficiently switch between these two representations as is proven in the following lemma.

\begin{lemma}[{Van der Hoeven and Schost~{\cite[Theorem~1]{DBLP:journals/aaecc/HoevenS13}}}]%
    \label{lem:interpolation}%
    Let $q$ be a constant prime power.
    Given an integer~$b\in\set{0,\dots,n}$ and an $n$-variate polynomial $P\in\field[X_1, \dots, X_n]$ with total degree at most~$\Delta$, we can compute the vector of evaluations $P(x)$ for all $x\in\TrimmedFp{n-b}{\Delta} \times \field^{b}$ in time $\OO(n\cdot \abs{\TrimmedFp{n-b}{\Delta}}\cdot q^b)$.
    Conversely, given $b$, and a vector of evaluations~$P(x)$ for all $x\in\TrimmedFp{n-b}{\Delta} \times \field^{b}$, we can compute the corresponding unique degree-$\Delta$ polynomial~$P$ in the same time.
\end{lemma}

The first algorithm in the lemma is called \emph{multipoint evaluation} and the second algorithm is called \emph{interpolation}.
The most natural cases of this lemma are $b=n$ and $b=0$, but we need the more general version for our algorithm.
The original formulation of the lemma in \cite[Theorem~1]{DBLP:journals/aaecc/HoevenS13} is even more general in that the set of evaluation points can be chosen more flexibly and the field~$\field$ can have super-constant size---this, however, affects the running time.

For $q=2$, Lemma~\ref{lem:interpolation} can be seen as the linear transformation over $\mathbb{F}_2$ that is known as the Möbius Transform.
Björklund et al. \cite{DBLP:journals/mst/BjorklundHKK10} showed that assuming a bounded-degree polynomial, there exists a Trimmed Möbius Transform that only requires a bounded number of evaluations to compute the coefficients of a polynomial (and vice versa).
For $q \ge 2$, there are several papers~\cite{DBLP:conf/stoc/Umans08,DBLP:journals/siamcomp/KedlayaU11,DBLP:journals/jacm/BhargavaGKM23,DBLP:journals/jacm/BhargavaGGKU24} that devised algorithms for multipoint evaluation (not interpolation) over $\field$, where the individual degree of each variable is bounded.
In our setting, where only the total degree of monomials is restricted, and the degree of individual variables can be up to $q-1$, these multipoint evaluation algorithms have exponential running time.
Instead we use a special case of a multipoint evaluation \cite[Theorem~1]{DBLP:journals/aaecc/HoevenS13} and interpolation \cite[Proposition~3]{DBLP:journals/aaecc/HoevenS13} algorithm (also see \cite[Theorem~4.4]{DBLP:journals/jc/HoevenL20}), because our algorithm requires to repeatedly switch between the two representations of a polynomial.

\subsection{Isolation Lemma and Low-degree Approximations.}

First, we state a version of the isolation lemma of Valiant and Vazirani~\cite{valiant1985np}.
This lemma allows us to probabilistically transform a polynomial equation system into an equivalent system that has at most one solution with high probability.
\begin{lemma}[Valiant and Vazirani~\cite{valiant1985np}]\label{lem:ValiantVazirani}
    Let $q$ be a prime power.
    There exists a randomized algorithm \Call{ValiantVazirani}{P_1,\dots,P_m} that for $n$-variate polynomials $P_1, \dots, P_m$ over $\field$ runs in time~$\OO(n^2)$
    and samples a uniformly random number~$\ell\in\set{0,\dots,n}$ of uniformly random $n$-variate linear functions $P_{m+1}, \dots, P_{m+\ell}$ that satisfy the following condition:
    \begin{itemize}
        \item \emph{(Uniqueness)} If there exists some $x\in\field^n$ with $P_i({x})=0$ for all $i\in[m]$, then with probability at least $\Omega(\frac1n)$, there exists exactly one~$x\in\field^n$ with $P_i({x})=0$ for all $i\in[m+\ell]$.
    \end{itemize}
\end{lemma}

Next, we approximate the polynomial equation system by one with a smaller number of polynomials, which are obtained probabilistically.

\begin{algorithm}[RazborovSmolensky]\label{algo:RazborovSmolensky}\upshape
    This algorithm receives as input $n$-variate degree-$d$ polynomials $P_1, \dots, P_m$ over $\field$, and a positive integer $\mu$. It outputs $\mu$ random linear combinations of the $P_j$'s.
    \begin{algorithmic}[1]
        \Function{RazborovSmolensky}{$P_1,\dots,P_m;\mu$}
        \For{$i = 1, \dots, \mu$}
            \State Independently and uniformly sample $m$ coefficients $\rho_{i,1}, \dots, \rho_{i,m}\in\field$.
            \State Let $\widetilde P_i(X_1,\dots,X_n) = \sum_{j=1}^m \rho_{i,j} \cdot P_j(X_1,\dots,X_n)$.\label{line:a}
        \EndFor
        \State\Return $\widetilde P_1, \dots, \widetilde P_\mu$
        \EndFunction
    \end{algorithmic}
\end{algorithm}

We state the resulting lemma as follows.
\begin{lemma}[\cite{razborov1987lower,smolensky1987algebraic}]\label{lem:RazborovSmolensky}
    \RazborovSmolensky{P_1, \dots, P_m; \mu} is a randomized algorithm that for $n$-variate polynomials $P_1, \dots, P_m$ over $\field$ runs in time $\OO(m\mu)\cdot\max_{i\in[m]}|P_i|$, where $|P_i|$ is the number of non-zero coefficients of $P_i$.
    For all ${x}\in\field^n$, the output $\widetilde P_1, \dots, \widetilde P_{\mu}$ consists of $n$-variate polynomials over $\field$ and satisfies the following:
    \begin{itemize}
        \item \emph{(Completeness)} If $P_i({x})=0$ holds for all $i\in[m]$, then $\widetilde P_j({x})=0$ holds for all $j\in[\mu]$.
        \item \emph{(Soundness)} If $P_i({x})\ne0$ holds for some $i\in[m]$, then with probability at least~$1-q^{-\mu}$, we have $\widetilde P_j({x})\ne0$ for some $j\in[\mu]$.
    \end{itemize}
    Moreover, each $\widetilde P_j$ has degree at most $d\coloneqq\max_{i\in[m]}\deg(P_i)$.
\end{lemma}

\section{Algorithm for Polynomial Equation Systems over Finite Fields}
\label{sec:sum-algorithm}

Recall that \PESqd{q}{d} is the problem of deciding whether some given degree-$d$ polynomials over~$\field$ have a common root, or equivalently, whether the corresponding polynomial equation system has a solution.
For convenience, we restate our main theorem here.
\PESalgorithm*
This section is dedicated to constructing the claimed algorithm and proving its claimed properties.
At the end of~Section~\ref{sec:running-time} we are finally in position to prove the theorem.
For the remainder of this section, we fix $q$ and~$d$ to be integer constants such that $q$ is a prime power and $d\ge1$.
For the proof, we follow the outline given in Section~\ref{sec:technical-overview}.

\subsection{From Equations to a Sum.}\label{sec:equations-to-sum}
We start by reducing the problem of determining whether a given polynomial equation system has a solution to that of computing the sum over all evaluation points of a particular polynomial over~$\field$.
For polynomials $P_1,\dots,P_m\in\field[X_1,\dots,X_n]$, we define the associated \emph{indicator polynomial $F\in\field[X_1,\dots,X_n]$} and the \emph{full sum} $Z\in\field$ as follows:
\begin{equation}\label{eq:counting-polynomial}
    F \coloneqq \prod_{i=1}^m \paren*{1-(P_i)^{q-1}}
    \quad
    \text{and}
    \quad
    Z \coloneqq \sum_{x\in\field^n} F(x)
    \,.
\end{equation}
By Fermat's little theorem, we have $F(x)\ne 0$ for a vector~$x\in\field^n$ if and only if $P_1(x)=\dots=P_m(x)=0$ holds.
Thus, if the polynomials do not have a common root, it is guaranteed that~$F$ is identically zero and the full sum satisfies $Z=0$.
However, the sum is taken over~$\field$ and may be zero also if the polynomials do have a common root. We avoid this situation by using the isolation lemma to ensure that $F(x)\ne 0$ holds for at most one vector $x\in\field^n$.
More formally, we use the following intermediate problem:
\defproblemtalt{\SUMqd{q}{d}}{Polynomials $P_1, \dots, P_m \in \field[X_1,\dots,X_n]$ of degree at most~$d$}{Compute the full sum~$Z$, where $Z\in\field$ is defined as in \eqref{eq:counting-polynomial}}

We use the isolation lemma to efficiently reduce from \PESqd{q}{d} to \SUMqd{q}{d}.
\begin{lemma}\label{lem:isolation-reduction}
    If \SUMqd{q}{d} can be computed in bounded-error randomized time $T(n,m)$, then \PESqd{q}{d} can be computed in bounded-error randomized time $\OO(nT(n,m+n))$.
\end{lemma}
\begin{proof}
    Let $A$ be a bounded-error randomized algorithm for \SUMqd{q}{d}.
    We use the isolation lemma, Lemma~\ref{lem:ValiantVazirani}, and add up to~$n$ random linear equations to the polynomial equation system. We then feed the at most $m+n$ resulting polynomials as input to~$A$. If the original system did not have any solutions, the new system does not have a solution either, and so~$A$ returns~$0$ with high probability.
    However, if the original system has at least one solution, then the new system has a unique solution with probability $\Omega(\tfrac1n)$, in which case~$A$ returns~$1$.
    By repeating this procedure and thus making $\OO(n)$ independent queries to~$A$, we can amplify the success probability to $99\%$ as required.
\end{proof}

\subsection{From a Sum to a Partial Sum.}\label{sec:sum-to-partial-sum}
In the following, we will describe how to compute \SUMqd{q}{d}.
One straightforward way to do this would be to use Multipoint Evaluation, that is, to evaluate~$F$ at \emph{all} points~$x\in\field^n$ and compute the sum.
Doing so would take time $\Omega(q^n)$ and provide no gains over exhaustive search.
The main idea for computing the full sum~$Z$ more efficiently is to iteratively compute polynomials~$Z_\beta$ that represent partial sums:
\begin{equation}\label{eq:def:partialsum}
    Z_\beta(Y_1,\dots,Y_{n-\beta}) \coloneqq \sum_{z\in\field^\beta} F(Y_1,\dots,Y_{n-\beta},z)
    \,.
\end{equation}
In this expression, the first $n-\beta$ variables remain untouched, and the sum is taken over all possible settings~$z$ for the last~$\beta$ variables.
In particular, we have $Z_0=F$ and $Z_n=Z$.
We remark that these partial sums can also be defined via partial derivatives of~$F$; for example, we have $\frac{\partial^{q-1}}{\partial x_n^{q-1}} F=(q-1)(q-2)\cdots 2\cdot Z_1$.
In Section~\ref{sec:PartialSums}, we describe a randomized algorithm~\PartialSum{} to compute $Z_\beta$, and in Section~\ref{sec:error-probability}, we prove the following lemma.
\begin{restatable}{lemma}{LemmaPartialSumError}%
    \label{lem:PartialSum-error}%
    For all $\beta\in\N$, for all~$m,n\in\N$ with $n\ge\beta$, and for all $n$-variate degree-$d$ polynomials $P_1,\dots,P_m$ over $\field$, the probability that $\PartialSum{P_1,\dots,P_m;\beta}$ returns a polynomial~$\widetilde Z_\beta$ with $\widetilde Z_\beta\ne Z_\beta$ is at most $\PartialSumError{n}$.
\end{restatable}
Given this lemma, we are ready to state \FullSum{} as a straightforward deterministic reduction to \PartialSum{}.
The algorithm has a parameter~$\kappa\in(0,1)$ that we will set to an optimal value later.%
\begin{algorithm}[FullSum]\label{algo:FullSum}\upshape
    This algorithm receives as input $n$-variate degree-$d$ polynomials $P_1, \dots, P_{m}$ over~$\field$, and depends on a global parameter $\kappa\in\R$ with $0<\kappa<\kappaUPPER$.
    It outputs an element $\widetilde Z\in\field$ that, with probability at least $1 - \PartialSumError{n}$, is equal to the full sum $Z$ defined in \eqref{eq:counting-polynomial}.
    \begin{algorithmic}[1]
    \Function{FullSum}{$P_1,\dots,P_m$}
        \State Set $\beta \gets \floor{\kappa n}$.

        \State Let $\widetilde Z_\beta$ be the polynomial returned by $\PartialSum{P_1,\dots,P_m;\beta}$.\label{line:fullsum:call}

        \State Use multipoint evaluation (Lemma~\ref{lem:interpolation}) to compute $\widetilde Z_\beta(y)$ for all $y\in\field^{n-\beta}$.\label{line:fullsum:evaluate}

        \State\Return $\widetilde Z \coloneqq \sum_{y \in \field^{n-\beta}} \widetilde Z_\beta(y)$\label{line:fullsum:return}
    \EndFunction
    \end{algorithmic}
\end{algorithm}%
\begin{lemma}\label{lem:FullSum}
    Let $\kappa$ be a real number with $0<\kappa<\kappaUPPER$.
    There is a randomized algorithm \FullSum{} that solves \SUMqd{q}{d} with error probability at most~$\PartialSumError{n}$ and in time
    \(\OO( T(m,n,\floor{\kappa n}) + q^{(1-\kappa) n} n )\),
    where $T(m,n,\beta)$ is the running time of $\PartialSum{P_1,\dots,P_m;\beta}$.
\end{lemma}
\begin{proof}
    By Lemma~\ref{lem:PartialSum-error}, with probability at least $1-\PartialSumError{n}$, the polynomial $\widetilde Z_\beta$ returned by \PartialSum{} satisfies
    \(\widetilde Z_\beta=Z_\beta\).
    Conditioned on this event, we have \(\widetilde Z=Z\) and \FullSum{} returns the correct value.
    For the running time, note that line~\ref{line:fullsum:call} of \FullSum{} takes time $T(m,n,\floor{\kappa n})$ and line~\ref{line:fullsum:evaluate} takes time~$\OO(q^{n-\kappa n} n)$ by Lemma~\ref{lem:interpolation}.
    This concludes the proof.
\end{proof}

\subsection{Algorithm for Partial Sums}\label{sec:PartialSums}
In this section, we describe the algorithm~\PartialSum{} for computing~$Z_\beta$.
We stress that the goal of this algorithm is to compute a representation of all monomial-coefficient pairs for which the coefficient is non-zero.
To get an algorithm that is more efficient than brute force, we need two key insights.
The first is that we can bound the degree of $Z_\beta$ from above as observed in the following lemma.

\begin{lemma}\label{lem:Zbeta-degree}
    Let $P_1,\dots,P_m$ be $n$-variate degree-$d$ polynomials.
    If~$\beta\in\set{0,\dots,n}$, then the partial sum polynomial~$Z_\beta$ has degree at most~$\Zdegree{m}{\beta}$, where
    \(
    \Zdegree{m}{\beta}
    \coloneqq(\min(md,n)-\beta)(q-1)
    \).
\end{lemma}
\begin{proof}
    The degree of $F$ is trivially at most $n(q-1)$, because $F$ is an $n$-variate polynomial over~$\field$. Moreover, the degree of~$F$ is at most $md(q-1)$, because we have $F = \prod_{i=1}^m (1- P_i^{q-1})$ and the degree of each $P_i$ is bounded by~$d$.
    Let $c_M$ for $M \in \{0, \dots, q-1\}^n$ be the coefficients of $F$ and define
    \[F_1(Y) \coloneqq \sum_{M \in \{0, \dots, q-1\}^{n-\beta} \times \{q-1\}^\beta} c_M \cdot \prod_{i=1}^{n-\beta} Y_i^{M_i}.\]
    By Lemma~\ref{lem:modqtrick}, we have $F_1 = (q-1)^\beta \cdot Z_\beta$, and so $F_1$ and $Z_\beta$ have the same degree.

    Consider any monomial $M = (M_1, \dots, M_{n-\beta})$ with non-zero coefficient in $F_1$.
    By definition of $F_1$, the monomial
    $M' = (M_1, \dots, M_{n-\beta}, {q-1, \dots, q-1})\in\set{0,\dots,q-1}^n$
    has a non-zero coefficient $c_{M'}$ in $F$.
    Since the degree of $M'$ is at most the degree of $F$, we obtain
    \[\deg M = \deg M' - (q-1)\beta \leq \deg F - (q-1)\beta \leq (\min(md,n)-\beta)(q-1)=\Zdegree{m}{\beta}\,.\]
    As this applies to every monomial~$M$ of $F_1$, the same bound applies to $\deg F_1$.
\end{proof}
The second key insight is that the degree of~$Z_\beta$ can be reduced to $\Zdegree{\beta+2}{\beta}$, by replacing the polynomials~$P_1,\dots,P_m$ with $\beta+2$ random polynomials~$\widetilde{P}_1,\dots,\widetilde{P}_{\beta+2}$ that are returned by \RazborovSmolensky{}.
Doing so will introduce quite a lot of errors, so our algorithm will have to call \RazborovSmolensky{} several times and correct these errors; in the following lemma, we analyze the errors incurred in the $j$-th call.

\begin{lemma}\label{lem:RazborovSmolensky-our-application}
    Let $j,\beta\in\N$, and let $P_1,\dots,P_m$ be $n$-variate degree-$d$ polynomials.
    Let $\widetilde P_{j,1},\dots,\widetilde P_{j,\beta+2}$ be the polynomials returned by $\RazborovSmolensky{P_1, \dots, P_m; \beta+2}$.
    We define the indicator and partial sum polynomials for $\widetilde P_{j,1},\dots,\widetilde P_{j,\beta+2}$ analogously to \eqref{eq:counting-polynomial} and~\eqref{eq:def:partialsum}:
    \begin{equation}\label{eq:approximate-F-Z-definition}
        F_{j}
        \coloneqq
        \prod_{i=1}^{\beta+2} (1-(\widetilde P_{j,i})^{q-1})
        \qquad
        Z_{\beta,j}(Y_1,\dots,Y_{n-\beta})\coloneqq\sum_{z\in\field^{\beta}} F_j(Y_1,\dots,Y_{n-\beta},z)
    \end{equation}
    Then for all $y\in\field^{n-\beta}$, we have
    \(
    \Pr\paren*{Z_{\beta,j}(y) \ne Z_{\beta}(y)} \le q^{-2}
    \).
\end{lemma}
\begin{proof}
    Combining soundness and completeness of Lemma~\ref{lem:RazborovSmolensky} with Fermat's little theorem directly implies \(
    \Pr\left(
    F_j(x)\ne F(x)
    \right) \le q^{-(\beta+2)}
    \) for all~$x\in\field^n$.
    We write $x=(y,z)$ with $y\in\field^{n-\beta}$ and $z\in\field^\beta$.
    For all $y\in\field^{n-\beta}$, we obtain the claim by a union bound:
    \[
        \Pr\left(
        Z_{\beta,j}(y)\ne Z_\beta(y)
        \right)
        \le
        \Pr\left(
        \exists z\in\field^\beta\colon
        F_j(y,z)\ne F(y,z)
        \right) \le q^\beta \cdot q^{-(\beta+2)} = q^{-2}\,.
    \]
    This completes the proof.
\end{proof}

Our algorithm \PartialSum{} to compute the partial sum $Z_\beta$ is laid out as Algorithm~\ref{algo:PartialSum}.
We use an additional parameter~$\lambda$ with $0\le\lambda\le1$ to tune the running time of the algorithm.
The base case of the algorithm is $\beta\le\lambda n$ or $n\le 3$, in which case it computes~$Z_\beta$ directly from its definition in~\eqref{eq:def:partialsum} using brute force.
Otherwise, we have $\beta > \lambda n$ and $n\ge 4$, and the algorithm aims to recursively compute $Z_{\beta-\lambda n}$, and then uses brute force to compute~$Z_{\beta}$ from $Z_{\beta-\lambda n}$. To compute~$Z_{\beta-\lambda n}$, it applies the process suggested by Lemma~\ref{lem:RazborovSmolensky-our-application}, and passes the polynomials $\widetilde P_{j,1},\dots,\widetilde P_{j,\beta+2}$ as input to the recursive call of \PartialSum{}. It repeats this process~$t$ times independently to obtain~$t$ possibly erroneous partial sum polynomials $\widetilde Z_{\beta - \lambda n}$. It then corrects all errors using plurality votes on all evaluation points of the~$t$ returned polynomials, thus reconstructing $Z_{\beta - \lambda n}$ with high probability.

\begin{algorithm}[PartialSum]\label{algo:PartialSum}\upshape
    This algorithm receives as input $n$-variate degree-$d$ polynomials $P_1, \dots, P_{m}$ over $\field$ and an integer~$\beta\ge0$, and it depends on a global parameter $\lambda\in\R$ with $0<\lambda\leq\kappa$.
    It outputs a polynomial~$\widetilde Z_{\beta}$ of degree at most~$\Zdegree{m}{\beta}$ that, with probability at least $1-\PartialSumError{n}$, is identical to $Z_{\beta}$ from \eqref{eq:def:partialsum}.
    \begin{algorithmic}[1]
    \Function{PartialSum}{$P_1,\dots,P_m,\beta$}
    \State Define the following parameters:
    \State\quad $t \coloneqq \lceil 96 n \ln q \rceil$ \Comment{---will make this many recursive calls.}
    \State\quad $\beta' \coloneqq \beta - \ceil{\lambda n}$ \Comment{---will use this value for $\beta$ in each recursive call.}
    \State\quad $\Delta \coloneqq \Zdegree{m}{\beta}$
    \label{line:Delta}%
    \Comment{---the degree of $Z_\beta$ by Lemma~\ref{lem:Zbeta-degree}.}%
    \State\quad $\Delta' \coloneqq \Zdegree{\beta'+2}{\beta'}$
    \Comment{---the degree of $Z_{\beta',j}$ by Lemma~\ref{lem:Zbeta-degree}.}%

    \If{$\beta<\ceil{\lambda n}$ \textup{or} $n\le 3$}
        \State Use Lemma~\ref{lem:interpolation} to compute the evaluations $P_1(y,z),\dots,P_m(y,z)$ for all $y\in\TrimmedFp{n-\beta}{\Delta}$ and all $z\in\field^\beta$.
        \label{line:base-case-evaluate}

        \State Use this information to compute the corresponding evaluations~$Z_\beta(y)$ via \eqref{eq:def:partialsum}.%
        \label{line:base-case-compute}

        \State Use Lemma~\ref{lem:interpolation} to interpolate $Z_\beta$ from these evaluations
        and \Return $Z_\beta$.%
        \label{line:base-case-return}
    \EndIf

    \For{$j \in \{1,\dots, t\}$\label{line:for-j}}
    \State Call~$\RazborovSmolensky{P_1, \dots, P_m; \beta'+2}$
    to obtain $\widetilde P_{j,1}, \dots, \widetilde P_{j,\beta'+2}$.%
    \label{line:razborov}

    \Comment{By Lemma~\ref{lem:RazborovSmolensky-our-application}, each evaluation of the partial sum polynomial $Z_{\beta',j}$ is equal to the corresponding evaluation of $Z_{\beta'}$ with probability at least $1-q^{-2}$.}

    \State Let $\widetilde Z_{\beta',j}$ be the polynomial returned by $\PartialSum{\widetilde P_{j,1}, \dots, \widetilde P_{j,\beta'+2}; \beta'}$.%
    \label{line:recursive-call}

    \Comment{This recursive call causes the error $\widetilde Z_{\beta',j} \ne Z_{\beta',j}$ with probability at most~$\PartialSumError{n}$.}

    \State Evaluate the polynomial $\widetilde Z_{\beta',j}$ of degree at most~$\Delta'$ on all points in $\TrimmedFp{n-\beta}{\Delta} \times \field^{\beta-\beta'}$ using Lemma~\ref{lem:interpolation}\label{line:evaluation}%
    \Comment{---this is possible by $\Delta'\le\Delta$.}%
    \EndFor

    \ForAll{$y \in \TrimmedFp{n-\beta}{\Delta}$}
        \ForAll{$u\in\field^{\beta-\beta'}$}
            \State Let $\widetilde Z_{\beta'}(y,u) \coloneqq \Call{Plurality}{\widetilde Z_{\beta',1}(y,u), \dots , \widetilde Z_{\beta',t}(y,u)}$.%
            \label{line:plurality}

            \Comment{The plurality vote selects the value that occurs most frequently, breaking ties arbitrarily. In the proof of Lemma~\ref{lem:PartialSum-error}, we show that this is very likely to correct all errors introduced in lines~\ref{line:razborov} and~\ref{line:recursive-call}.}
        \EndFor
        \State Let $\widetilde Z_{\beta}(y) \coloneqq \sum_{u \in \field^{\beta-\beta'}} \widetilde Z_{\beta'}(y,u)$.%
        \label{line:Zsum}
    \EndFor

    \State Interpolate the polynomial~$\widetilde Z_{\beta}$ from its evaluations on all~$y\in\TrimmedFp{n-\beta}{\Delta}$ using Lemma~\ref{lem:interpolation}.%
    \label{line:interpolation}

    \State\Return $\widetilde Z_{\beta}$.\label{line:return}
    \EndFunction
    \end{algorithmic}
\end{algorithm}

\subsection{Error Probability of the Algorithm for Partial Sum.}\label{sec:error-probability}

In this section, we prove Lemma~\ref{lem:PartialSum-error}, our bound on the error probability of \PartialSum{}.

\LemmaPartialSumError*
\begin{proof}
    We prove the claim by induction on~$\beta$, so let $\beta\in\N$.
    Let $P_1,\dots,P_m$ be the given degree-$d$ polynomials, let $Z_\beta$ be the partial sum polynomial defined in~\eqref{eq:def:partialsum}.
    In the base case, we have $\beta\le\ceil{\lambda n}$ or $n\le 3$, and the polynomial that is returned in line~\ref{line:base-case-return} is equal to the partial sum polynomial~$Z_\beta$, because $Z_\beta$ has degree at most~$\Delta$ by Lemma~\ref{lem:Zbeta-degree} and hence the interpolation is able to recover all non-zero coefficients of $Z_\beta$ by Lemma~\ref{lem:interpolation}.
    For the inductive case, suppose we have $\beta>\ceil{\lambda n}$ and $n\ge 4$, and that the claim is true for all~$\beta'$ with $0\le\beta'<\beta$.
    Let $\widetilde{Z}_\beta$ be the polynomial constructed in line \ref{line:interpolation}.
    We need to prove $\Pr\paren[\big]{\widetilde{Z}_\beta\ne Z_\beta}\le q^{-n}$.

    We start by analyzing the effects of line~\ref{line:razborov} on the partial sum polynomial~$Z_{\beta',j}$ defined from $\widetilde{P}_{j,1},\dots,\widetilde{P}_{j,\beta'+2}$ in~\eqref{eq:approximate-F-Z-definition}.
    We apply Lemma~\ref{lem:RazborovSmolensky-our-application} with~$\beta'$ in place of~$\beta$ and get:
    \begin{equation}\label{eq:prob-y-u-evaluation-differ}
        \forall j\in\set{1,\dots,t}\;\;
        \forall y\in\field^{\beta}\;\;
        \forall u\in\field^{\beta-\beta'}
        \colon
        \Pr(Z_{\beta',j}(y,u) \ne Z_{\beta'}(y,u))
        \le
        q^{-2}\,.
    \end{equation}

    In line~\ref{line:recursive-call}, the algorithm makes a recursive call to $\PartialSum{\widetilde{P}_{j,1},\dots,\widetilde{P}_{j,\beta'+2};\beta'}$, which returns a polynomial $\widetilde{Z}_{\beta',j}$ that is supposed to be identical with $Z_{\beta',j}$.
    Indeed, since we have $\beta'<\beta$, the induction hypothesis implies
    \begin{equation}\label{eq:prob-recursion-fails}
        \forall j\in\set{1,\dots,t}
        \colon
        \Pr(\widetilde{Z}_{\beta',j}\ne Z_{\beta',j})
        \le q^{-n}
        \,.
    \end{equation}

    We now argue that if for all $y \in \TrimmedFp{n-\beta}{\Delta}$ and $u \in \field^{\beta-\beta'}$, the event $\widetilde{Z}_{\beta'}(y,u) = Z_{\beta'}(y,u)$ occurs, then we also have $\widetilde{Z}_\beta = Z_\beta$.
    To prove this, assume that the former holds.
    By construction of $\widetilde{Z}_\beta$ and definition of $Z_\beta$ and $Z_{\beta'}$, we have for all $y \in \TrimmedFp{n-\beta}{\Delta}$:
    \[\widetilde{Z}_\beta(y) = \sum_{u \in \field^{\beta-\beta'}} \widetilde{Z}_{\beta'}(y,u) = \sum_{u \in \field^{\beta-\beta'}} Z_{\beta'}(y,u).\]    
    Since $\widetilde{Z}_\beta$ is constructed in line \ref{line:interpolation} by interpolating from evaluations on $y\in\TrimmedFp{n-\beta}{\Delta}$, the polynomial $\widetilde{Z}_\beta$ must have degree at most~$\Delta$ by~Lemma~\ref{lem:interpolation}.
    Again by~Lemma~\ref{lem:interpolation}, as the two degree-$\Delta$ polynomials $\widetilde{Z}_\beta$ and $Z_\beta$ agree on all points in $\TrimmedFp{n-\beta}{\Delta}$, they must be identical.
    This proves the claim, and we also obtain the contrapositive, that is:
    If $\widetilde{Z}_\beta \ne Z_\beta$, there exist~$y\in\TrimmedFp{n-\beta}{\Delta}$ and $u\in\field^{\beta-\beta'}$ such that the event $\widetilde{Z}_{\beta'}(y,u)\ne Z_{\beta'}(y,u)$ occurs.
    We will prove the following regarding the probability of that event:
    \begin{equation}\label{eq:forall-y-u}
        \text{For all $y\in\TrimmedFp{n-\beta}{\Delta}$ and $u\in\field^{\beta-\beta'}$, we have \(\Pr\paren[\big]{\widetilde{Z}_{\beta'}(y,u)\ne Z_{\beta'}(y,u)}\le q^{-2n}\)}\,.
    \end{equation}
    Once this is established, a union bound over the at most~$q^n$ pairs~$(y,u)$ leads to the claimed final bound:
    \[
        \Pr\paren[\big]{\widetilde{Z}_\beta \ne Z_\beta}
        \le
        \Pr\paren[\big]{\exists y,u\colon\widetilde{Z}_{\beta'}(y,u)\ne Z_{\beta'}(y,u)}\le q^{n-\beta'}\cdot q^{-2n}\le q^{-n}
        \,.
    \]

    It remains to prove \eqref{eq:forall-y-u}, so let~$y\in\TrimmedFp{n-\beta}{\Delta}$ and $u\in\field^{\beta-\beta'}$ be arbitrary.
    Recall that the value $\widetilde{Z}_{\beta'}(y,u)$ is constructed in line~\ref{line:plurality} by a plurality vote over the evaluations $\widetilde{Z}_{\beta',1}(y,u),\dots,\widetilde{Z}_{\beta',t}(y,u)$.
    For each~$j\in\set{1,\dots,t}$, we define the random variable~$X_j$ that indicates whether the $j$-th value $\widetilde{Z}_{\beta',j}(y,u)$ in this plurality vote was computed correctly in lines~\ref{line:razborov} and~\ref{line:recursive-call}, that is, we have
    \[
        X_j =
        \begin{cases}
            1, & \text{if $\widetilde{Z}_{\beta',j}(y,u)=Z_{\beta'}(y,u)$;} \\
            0, & \text{otherwise.}
        \end{cases}
    \]
    Moreover, let $X\coloneqq\sum_{j=1}^t X_j$.
    If the plurality fails to produce the correct result~$Z_{\beta'}(y,u)$, then the event $X\le t/2$ must occur.
    By \eqref{eq:prob-y-u-evaluation-differ} and \eqref{eq:prob-recursion-fails}, we have $\Pr(X_j=0)\le q^{-2} + q^{-n}\le 1/3$ when $q \geq 2$ and $n\ge 4$, and thus $\E(X)\ge \tfrac23 t$.
    Since the random variables are independent and identically distributed, we can apply the Chernoff bound stated in Lemma~\ref{lem:chernoff} with $\delta=\tfrac14$:
    \begin{align*}
        \Pr\Big( X \le \tfrac t2 \Big)
         & \le
        \Pr\Big( X \le (1 - \tfrac14) \, \E(X) \Big)
        \le
        \exp\Big(- \tfrac{\tfrac{1}{16} \E(X)}{2}\Big)
        \\
         & \leq
        \exp\Big(- \tfrac {t}{48}\Big)
        =
        \exp\Big(- \tfrac {\lceil 96 n\ln q \rceil}{48}\Big)
        \leq q^{-2n}
        \,.
    \end{align*}
    We obtain $\Pr\paren[\big]{\widetilde{Z}_{\beta'}(y,u)\ne Z_{\beta'}(y,u)}\le q^{-2n}$ and \eqref{eq:forall-y-u} follows. This concludes the proof.
\end{proof}

\subsection{Running Time of the Algorithm for Partial Sum.}\label{sec:running-time}
In this section, we prove an upper bound on the running time of~\PartialSum{}.
Let $T(m,n,\beta)$ be the worst-case running time of \PartialSum{} when the input consists of {$n$-variate} polynomials $P_1,\dots,P_m$ and parameter~$\beta$.
\FullSum{} uses an initial value of $\floor{\kappa n}$ for~$\beta$, where $\kappa < 1/(2d-1)$, which is why we only consider this setting of $\beta$ in the following lemma.
\begin{lemma}\label{lem:PartialSum-time}%
    Let $\kappa,\lambda\in(0,1)$ be real constants with $0<\lambda\leq\kappa<\kappaUPPER$.
    For all positive integers~$D$, let $\Delta_D\coloneqq(\floor{\kappa n}-D\ceil{\lambda n}) (d-1)(q-1) + 2d(q-1)$.
    Then the running time~$T(m,n,\floor{\kappa n})$ of \PartialSum{} satisfies
    \begin{align}\label{eq:T-bound}
        T(m,n,\floor{\kappa n})&\leq\OOstar\paren[\Big]{
        \max\setc*{n^D T(D)}{D\in\N\text{ and }D\ceil{\lambda n}\leq\floor{\kappa n}}
        }
    \end{align}
    for a function $T(D)$ with
    \begin{align}\label{eq:T(0)-bound}
    T(0)
    &\leq
    \OOstar(q^{n-\floor{\kappa n}+\ceil{\lambda n}})\,\text{ and}\\
    \label{eq:T(D)-bound}
    T(D)
    &\leq
    \OOstar\paren*{\binom{n-\floor{\kappa n}+D\ceil{\lambda n}}{\Delta_D}_q \cdot q^{\ceil{\lambda n}}}
    \text{ for all $D>0$}\,.
    \end{align}
\end{lemma}
\begin{proof}
    If $n$ is at most a constant, then the claim is trivial, so we can assume without loss of generality that $n$ is large enough, such that $\ceil{\lambda n}\leq\floor{\kappa n}$ holds.
    In order to bound $T(m,n,\floor{\kappa n})$, we consider the recursion tree of \PartialSum{}.
    We introduce some notation:
\begin{itemize}
    \item
    Let $\beta_D$ be the value of~$\beta$ at depth~$D$, so at the root we have $\beta_0=\floor{\kappa n}$.
    Each recursive call subtracts $\ceil{\lambda n}$ from $\beta$.
    Thus, we have $\beta_D=\floor{\kappa n} - D \ceil{\lambda n}$.
    \item
    Let~$D^\ast$ be the depth of the recursion tree. The leaves are reached when $0\le\beta_D<\ceil{\lambda n}$ holds, which is equivalent to
    $D \ceil{\lambda n} \leq \floor{\kappa n}<(D+1) \ceil{\lambda n}$.
    Thus, $D^\ast$ is the largest integer~$D$ with~$D \ceil{\lambda n} \leq \floor{\kappa n}$ and thereby bounded by a constant.
    \item
    Let $m_D$ be the number of polynomials at depth~$D$. Then $m_0$ is the initial number $m$ of polynomials and $m_D=\beta_D+2$ holds for all $D>0$.
    \item
    Let $\Delta_D$ be the value of~$\Delta$ at depth~$D$.
    By definition of $\Zdegree{m}{\beta}$ in Lemma~\ref{lem:Zbeta-degree}, we have
    $\Delta_D=\Zdegree{m_D}{\beta_D}=(\min(m_D d,n)-\beta_D)(q-1)$.
    If $D>0$, then by $\lambda\le\kappa<\kappaUPPER$, we have $m_Dd\le n$ for large enough~$n$, and so $\Delta_D=(m_Dd-\beta_D)(q-1)$, which can be easily verified to coincide with the definition of $\Delta_D$ in the lemma statement.
    Moreover, by $\beta_D \leq \floor{\kappa n} < n/(2d-1)$, we have $\Delta_D\leq n(d-1)(q-1)/(2d-1)+2d(q-1) < n(q-1)/2$ for large enough~$n$.
    For $D=0$, we have $m_0d\ge n$ without loss of generality, and so $\Delta_0=(n - \floor{\kappa n})(q-1)$.
    We further assume without loss of generality that~$n$ is at least a large enough constant depending only on~$q,d,\kappa,\lambda$, so that~$\Delta_D$ is decreasing with~$D$, that is, we have $\Delta_0\geq\Delta_1\geq\dots\geq\Delta_{D^\ast}$.
\end{itemize}

In order to bound $T(m,n,\floor{\kappa n})$ as in \eqref{eq:T-bound}, let $T(D)$ be the running time contribution of a single node at level~$D$ in the recursion tree of \PartialSum{}.
Since $D^\ast$ is the depth of the recursion tree, we have $0\leq D\leq D^\ast$.
Moreover, each non-leaf of the tree has exactly~$t$ children, thus the number of nodes at depth~$D$ is equal to~$t^D$.
Since $D^\ast$ is the largest integer~$D$ that satisfies $D\ceil{\lambda n}\leq\floor{\kappa n}$ and $t \in \OO(n)$, the bound in \eqref{eq:T-bound} follows from the definition of~$T(D)$.

To prove \eqref{eq:T(0)-bound} and \eqref{eq:T(D)-bound}, we distinguish the base case ($D=D^\ast$) and the recursive cases ($0\leq D<D^\ast$) of the recursion tree of \PartialSum{}.

\emph{Base Case ($D=D^\ast$).}
Recall that $D^\ast>0$ holds. We claim that \eqref{eq:T(D)-bound} holds for $D=D^\ast$.
The leaves of the recursion tree of \PartialSum{} are at depth $D^\ast$ and correspond to the base case of \PartialSum{}, that is, Lines~\ref{line:base-case-evaluate}--\ref{line:base-case-return}.
These lines are only executed if $n\leq 3$ or $\beta<\ceil{\lambda n}$ holds.
If $n\le 3$, then the algorithm takes constant time, so we assume $n\ge 4$ and $\beta=\beta_{D^\ast}<\ceil{\lambda n}$ without loss of generality.
Writing $\Delta\coloneqq\Delta_{D^\ast}$, we analyze the running time~$T(D^\ast)$ as follows:
\begin{itemize}
    \item Line~\ref{line:base-case-evaluate} takes time $\OO(\binom{n-\beta}{\downarrow\Delta}_q\cdot q^\beta\cdot n\cdot m)$ by Lemma~\ref{lem:interpolation}.
    \item Line~\ref{line:base-case-compute} takes time $\OO(q^\beta m)$ for each of the~$\binom{n-\beta}{\downarrow\Delta}_q$ evaluation points~$y$.
    \item Line~\ref{line:base-case-return} takes time $\OO(\binom{n-\beta}{\downarrow\Delta}_q n)$, again by Lemma~\ref{lem:interpolation}.
\end{itemize}
Thus, the running time of the base case is dominated by Line~\ref{line:base-case-evaluate}.
By monotonicity of the extended binomial coefficient in~$\Delta$ for $\Delta\leq n(q-1)/2$, we have $\binom{n-\beta}{\downarrow\Delta}_q\leq \Delta\binom{n-\beta}{\Delta}_q$.
This establishes the running time bound \eqref{eq:T(D)-bound} for $D=D^\ast$.

\emph{Recursive Case.}
The non-leaves of the recursion tree of \PartialSum{} occur at depth~$D$ for $0 \leq D < D^\ast$ and correspond to the recursive case of the algorithm, that is, Lines~\ref{line:for-j}--\ref{line:return}.
We remark that $\beta=\beta_D$ and $\beta'=\beta_{D+1}$ hold at depth~$D$.
In order to show~\eqref{eq:T(0)-bound} for $D=0$ and~\eqref{eq:T(D)-bound} for $0<D<D^\ast$, we consider the running time contribution of each line of the recursive case:
\begin{itemize}
    \item Line~\ref{line:razborov} is executed~$t$ times and calls~\RazborovSmolensky{}, the running time of which is stated in Lemma~\ref{lem:RazborovSmolensky}.
    This leads to a contribution of $\OO(t m_D \beta' \binom{n}{d}_q)\le\OOstar(\binom{n}{\Delta}_q)$ by $d\le \Delta$.
    \item Line~\ref{line:recursive-call} issues a total of~$t$ recursive calls, and thus contributes~$\OOstar(1)$ to the running time~$T(D)$ of any node at level~$D$.
    \item Line~\ref{line:evaluation} causes a total of~$t=\OO(n)$ calls to Lemma~\ref{lem:interpolation} with $\max(\Delta,\Delta')$ in place of~``$\Delta$'', but we can use $\Delta'\leq\Delta$ to simplify the expression. Thus, Line~\ref{line:evaluation} contributes
    $\OO(tn\binom{n-\beta}{\downarrow\Delta}_q \cdot q^{\beta-\beta'})$ to the running time.
    \item Each execution of Line~\ref{line:plurality} contributes $\OO(t)\le\OO(n)$, leading to a total contribution of
          $\OO(n \binom{n-\beta}{\downarrow\Delta}_q \cdot q^{\beta-\beta'})$.
    \item The contribution of Lines~\ref{line:Zsum} and~\ref{line:interpolation} is dominated by the one of Line~\ref{line:plurality}; for Line~\ref{line:Zsum} this is trivial and for Line~\ref{line:interpolation} this follows from Lemma~\ref{lem:interpolation}.
\end{itemize}
The running time is dominated by Line~\ref{line:evaluation},
which with $\binom{n-\beta}{\downarrow\Delta}_q\leq \Delta\binom{n-\beta}{\Delta}_q$ establishes~\eqref{eq:T(D)-bound} for $0<D<D^\ast$.
To prove \eqref{eq:T(0)-bound} for $D=0$, we note $\binom{n-\beta}{\downarrow\Delta}_q \cdot q^{\beta-\beta'} \leq q^{n-\beta+\beta-\beta'}\leq q^{n-\floor{\kappa n} + \ceil{\lambda n}}$.

In summary, we have established \eqref{eq:T(0)-bound} for $D=0$ and \eqref{eq:T(D)-bound} for all $D$ with $0<D\leq D^\ast$, and thus \eqref{eq:T-bound} follows as claimed.
\end{proof}

\begin{proof}(Proof of Theorem~\ref{thm:pes-algorithm})
    We construct our algorithm for \PESqd{q}{d} as follows:
    First, we observe that \FullSum{} of Lemma~\ref{lem:FullSum} is a bounded-error randomized algorithm for \SUMqd{q}{d}.
    By Lemma~\ref{lem:isolation-reduction}, we can thus construct a bounded-error randomized algorithm for \PESqd{q}{d} whose running time is the running time of \FullSum{} times~$n$; since we do not care about polynomial factors here, it thus remains to bound the running time of \FullSum{}.

    Let $0<\lambda\leq\kappa < \frac{1}{2d-1}$. We claim that the running time of \FullSum{} is at most $\OOstar(q^{(\zeta_{q,d}(\kappa)+\lambda) \cdot n})$, where we define $\zeta_{q,d}(\kappa)$ as follows:
    \begin{align}\label{def:zetaqd-kappa}
        \zeta_{q,d}(\kappa)\coloneqq\max\set[\Big]{1-\kappa,\;\sup_{0\leq\delta\leq\kappa} H(q,\alpha)\cdot(1-\delta)}\,,\text{ where $\alpha\coloneqq\frac{\delta(d-1)}{1-\delta}$}
        \,.
    \end{align}
    Then since $\lambda>0$ can be an arbitrarily small constant, setting $\zeta_{q,d}$ to be any constant bigger than $\inf_{0<\kappa < 1/(2d-1)} \zeta_{q,d}(\kappa)$ gives the exponent in the running time of \FullSum{}.
    It remains to prove the claim on $\zeta_{q,d}(\kappa)$.
    Recall that \FullSum{} calls $\PartialSum{P_1,\dots,P_m;\floor{\kappa n}}$ and that $T(m,n,\floor{\kappa n})$ is the running time of \PartialSum{}.
    By Lemma~\ref{lem:FullSum}, the running time of \FullSum{} is $\OOstar(q^{(1-\kappa)n} + T(m,n,\floor{\kappa n}))$.
    By $1-\kappa \leq \zeta_{q,d}(\kappa)$, it remains to bound~$T(m,n,\floor{\kappa n})$.

    Since the algorithms $\FullSum{}$ and $\PartialSum{}$ as well as the running time bounds in Lemma~\ref{lem:PartialSum-time} do not depend on the precise values of $\kappa$ and $\lambda$, but only on the rounded values $\floor{\kappa n}$ and $\ceil{\lambda n}$, we can assume without loss of generality that $\kappa n$ and $\lambda n$ are integers, which simplifies notation.
    Moreover, we will choose $\kappa,\lambda\geq0$ such that $D^\ast=\kappa/\lambda$ is an integer.
    By~\eqref{eq:T-bound}, the running time of \PartialSum{} is at most $\OOstar(\max\setc{T(D)}{0\leq D\leq \kappa/\lambda})$.
    By \eqref{eq:T(0)-bound}, we have $T(0)=\OOstar(q^{(1-\kappa+\lambda)n})\leq\OOstar(q^{(\zeta_{q,d}(\kappa)+\lambda)n})$, so it remains to prove this bound for $D>0$.

    Indeed, if we write $\delta\coloneqq\kappa-D\lambda$ and $\widetilde\Delta\coloneqq n \delta (d-1)(q-1)+2d(q-1)$, we have
    \[
        T(D)\leq\OOstar\paren*{\binom{n(1-\delta)}{\widetilde\Delta}_q \cdot q^{\lambda n}}\text{for all $D>0$.}
    \]
    Now our assumption $0\leq\delta\leq\kappa<\kappaUPPER$ implies $\delta(d-1) < (1-\delta)/2$, and thus, if $n$ is at least a large enough constant, we have
    $\widetilde\Delta < n(1-\delta)(q-1)/2$. Thus, $\widetilde\Delta$ is in the increasing part of the extended binomial coefficient, that is, $\binom{n(1-\delta)}{\Delta'}_q < \binom{n(1-\delta)}{\widetilde\Delta}_q$ holds for all $\Delta'$ with $\Delta'<\widetilde\Delta$.
    On the other hand, we have $\binom{n(1-\delta)}{\Delta'-1}_q \geq \Omega(\tfrac{1}{n}\cdot\binom{n(1-\delta)}{\Delta'}_q)$, and so we can ignore the constant $2d(q-1)$ term in~$\widetilde\Delta$, since this term affects only the polynomial factors of the asymptotics.

    Let $\Delta\coloneqq n\delta (d-1)(q-1)=\alpha(q-1)(1-\delta)n$ and $\alpha\coloneqq \delta(d-1)/(1-\delta)$.
    By Lemma~\ref{lem:ext-binom-bound} (applied with $(1-\delta)n$ in place of $n$), we thus have: 
    \begin{align*}
        \binom{n(1-\delta)}{\Delta}_q \cdot q^{\lambda n}
        &\leq
        q^{H(q,\alpha)\cdot(1-\delta)n + \lambda n}
        \leq
        q^{(\zeta_{q,d}(\kappa)+\lambda)\cdot n}
        \,.
    \end{align*}
    Thus, for small enough $\lambda>0$ and the best choice of $\kappa$, the running time of \FullSum{} is at most $\OOstar(q^{\zeta_{q,d} n})$ as claimed.

Finally, we show the claimed bound $\zeta_{q,d}\leq 1 - \min\paren*{\tfrac{1}{8\ln q}, \tfrac{1}{4d}}$ on the exponent.
To this end, we use the definition of $\zeta_{q,d}(\kappa)$ in~\eqref{def:zetaqd-kappa}.
Let $\kappa=\tfrac{1}{4d}$.
Then $1-\kappa=1-\tfrac{1}{4d}$, and so it remains to bound the second term in the definition of $\zeta_{q,d}(\kappa)$.
We bound that term as follows:
\begin{align*}
    \sup_{0\leq\delta\leq\kappa} H(q,\alpha)\cdot(1-\delta) &
    \leq
    \sup_{0\leq\delta\leq\kappa} H(q,\alpha)
    \\
    \intertext{Since $\alpha\coloneqq\alpha(\delta)\coloneqq\frac{\delta(d-1)}{1-\delta}<\tfrac12$ is increasing in $\delta$ for $\delta\in[0,\kappa]$ and $H(q,\tilde\alpha)$ is increasing in~$\tilde\alpha$ for $\tilde\alpha\in[0,\tfrac12]$, we know that the supremum is attained at $\delta=\kappa$:}
    &
    = H\paren*{q,\alpha(\kappa)}\\
    \intertext{Moreover, we have $\alpha(\kappa)=\frac{\kappa (d-1)}{1-\kappa}=\frac{(d-1)}{(4d)(1-1/(4d))} < \frac14$, and thus we can bound the entropy as follows:}
    &< H(q,\tfrac14)=1-{I(q-1,\tfrac14)}/{\ln q}\,.
\end{align*}
Numerically, we see 
$I(q-1,\tfrac14)\approx 0.1308\geq\tfrac{1}{8}$ for $q=2$ and by Lemma~\ref{lem:exponent-properties}, $I(q-1,\tfrac14)$ is increasing in $q$.
This proves the claim on the exponent and concludes the proof.
\end{proof}
We remark that, in the final calculation, the limit satisfies
$\lim_{q\to\infty} I(q-1,\tfrac14)=I^\ast_{1/4}\approx 0.408639$, so our bound remains of the form $1-c/\ln q$ even for large~$q$. 

\section{Conditional Lower Bounds}

Our conditional hardness result of Theorem~\ref{thm:pes-hardness} relies on a hypothesis by Impagliazzo and Paturi~\cite{impagliazzo2001complexity}.

\begin{GrayBox}{Strong Exponential Time Hypothesis (\pp{SETH})}
    For all $\eps > 0$, there is some $k \ge 3$ such that \pp{$k$-SAT} cannot be solved in time $\OO(2^{(1-\eps) n})$.
\end{GrayBox}

In order to prove Theorem~\ref{thm:pes-hardness}, we devise a suitable mapping reduction from \pp{$k$-SAT} to \PESqd{q}{d} for all fixed~$k$.
Recall that a \emph{mapping reduction from $L\subseteq\zo^\ast$ to $L'\subseteq\zo^\ast$} is an algorithm~$f\colon\zo^\ast\to\zo^\ast$ that satisfies $x\in L$ if and only if $f(x)\in L'$.
Moreover, a mapping reduction is \emph{parsimonious} if it preserves the number of solutions.
It should be noted that we need a somewhat precise bound on the number of variables and cannot just use $\OO$-notation, as will become clear in the proof of Theorem~\ref{thm:pes-hardness}.
\begin{lemma}[Reduction from \pp{$k$-SAT} to \PESqd{q}{d}]\label{lem:PESreduction}
    Let $k \in \N$, $q$ be a prime power and let $\delta>0$ be rational.
    There is a parsimonious mapping reduction from \pp{$k$-SAT} to \PESqd{q}{d} that is given an $n$-variable \pp{$k$-CNF} formula with $m$ clauses and produces a polynomial equation system over~$\field$ with at most $\frac{n}{\log q} \cdot \left( 1 + \frac{\delta}{2} + o(1) \right)$ variables,
    $m$ equations, and degree at most $k \cdot (\frac{2}{\delta} + \frac{1}{\log q} + 1) \cdot (q-1) \in \OO(1)$.
    Moreover, this reduction runs in time $\OO(nm)$.
\end{lemma}
\begin{proof}
    We first show that there is a mapping reduction from \pp{$k$-SAT} to \PESqd{q}{d} with the desired properties.
    In the end, we argue that this reduction can also be made parsimonious.
    Let $\FF = C_1 \land \dots \land C_m$ be a propositional formula in $k$-CNF and let $C_i = \ell_{i,1} \lor \dots \lor \ell_{i,k}$ for all $i \in [m]$.
    We construct a polynomial equation system $E$ such that $E$ has a solution if and only if $\FF$ is satisfiable.
    The idea is to encode blocks of Boolean variables by blocks of variables over $\field$ and use interpolation to obtain polynomials that decode these blocks and output the Boolean value of individual variables.
    Then, we construct polynomials $P_i$ for the clauses of~$\FF$ such that, for all $i \in [m]$, the polynomial~$P_i$ is $\{0,1\}$-valued and evaluates to $1$ on an assignment~$\hat X$ over~$\field$ if and only if the Boolean assignment encoded by $\hat X$ satisfies $C_i$.

    We now give the details.
    Let $\mathsf{vars}_1 \coloneqq \lceil \frac{2}{\delta} \cdot \log q \rceil$ and $\mathsf{blocks} \coloneqq \lceil \frac{n}{\mathsf{vars}_1} \rceil$.
    For simplicity of the construction, we assume that $\FF$ has exactly $\mathsf{blocks} \cdot \mathsf{vars}_1$ variables by introducing additional dummy variables if necessary.
    We will still consider $n$ to be the original number of variables, allowing us to precisely analyze the parameters of the reduction.
    Furthermore, we assume without loss of generality that the variables in $\FF$ are named $x = x_{1,1}, \dots, x_{1,\mathsf{vars}_1}, \dots, x_{\mathsf{blocks},1}, \dots, x_{\mathsf{blocks},\mathsf{vars}_1}$, that is, the variables are grouped into $\mathsf{blocks}$ many blocks of $\mathsf{vars}_1$ variables each.
    Each block of $\mathsf{vars}_1$ Boolean variables is now encoded by a block of variables over $\field$ of suitable arity.
    This arity will be $\mathsf{vars}_2 \coloneqq \lceil \frac{\mathsf{vars}_1}{\log q} \rceil$, which is sufficient to encode all assignments to the block of Boolean variables, as we have
    \[q^{\mathsf{vars}_2} = 2^{\log q  \cdot \mathsf{vars}_2} \geq 2^{\log q \cdot \frac{\mathsf{vars}_1}{\log q}} = 2^{\mathsf{vars}_1}.\]
    Now, fix an encoding by choosing any efficiently computable surjective $\{0,1\}$-valued function
    \[\mathrm{dec} \colon \field^{\mathsf{vars}_2} \to \set{0,1}^{\mathsf{vars}_1}.\]
    For example, $\mathrm{dec}$ can be chosen as the function mapping any tuple $X \in \field^{\mathsf{vars}_2}$ to the binary encoding of the number represented by $X$ when interpreted as a base-$q$ number, modulo $2^{\mathsf{vars}_1}$.
    For all $v_1 \in [\mathsf{vars}_1]$, let $\mathrm{DEC}_{v_1}$ be the $\mathsf{vars}_2$-variate polynomial over $\field$ that agrees with the $v_1$-th bit of $\mathrm{dec}$ on all inputs.

    To express the constraints imposed by the clauses $C_i$ using polynomial equations, define for all $i \in [m]$ and $j \in [k]$ the polynomial
    \[L_{i,j}(Y) =
        \begin{cases}
            Y,   & \text{if } \ell_{i,j} \text{ is a positive literal} \\
            1-Y, & \text{otherwise}.
        \end{cases}
    \]
    We now construct the desired polynomial equation system $E$ over $\field$.
    The variables of $E$ will be $X = X_{1,1}, \dots, X_{1,\mathsf{vars}_2}, \dots, X_{\mathsf{blocks},1}, \dots, X_{\mathsf{blocks},\mathsf{vars}_2}$, where the block $X_{b,1}, \dots, X_{b,\mathsf{vars}_2}$ encodes the block $x_{b,1}, \dots, x_{b,\mathsf{vars}_1}$ of Boolean variables for all $b \in [\mathsf{blocks}]$.
    For any $i \in [m]$ and $j \in [k]$, let $b(i,j)$ be the index of the block of the variable occurring in the literal $\ell_{i,j}$ and let $v_1(i,j)$ be its position inside that block.
    For example, if $\ell_{i,j} = \neg x_{1,3}$, then $b(i,j) = 1$ and $v_1(i,j) = 3$.
    Consider the polynomial
    \[Q_{i,j}(X) \coloneqq L_{i,j}(\mathrm{DEC}_{v_1(i,j)}(X_{b(i,j), 1}, \dots, X_{b(i,j), \mathsf{vars}_2})).\]
    This polynomial uses $\mathrm{DEC}_{v_1(i,j)}$ to obtain the value of the Boolean variable $x_{b(i,j), v_1(i,j)}$ in the Boolean assignment encoded by the variables $X_{1,1}, \dots, X_{1,\mathsf{vars}_2}, \dots, X_{\mathsf{blocks},1}, \dots, X_{\mathsf{blocks}, \mathsf{vars}_2}$ and possibly negates it depending on whether $\ell_{i,j}$ is positive or negative.
    To make this formal, fix an assignment $\hat X = \hat X_{1,1}, \dots, \hat X_{1,\mathsf{vars}_2}, \dots, \hat X_{\mathsf{blocks},1}, \dots, \hat X_{\mathsf{blocks},\mathsf{vars}_2} \in \field^{\mathsf{blocks}\cdot \mathsf{vars}_2}$.
    Then $\mathrm{dec}(\hat X) \coloneqq \mathrm{dec}(\hat X_{1,1}, \dots, \hat X_{1,\mathsf{vars}_2}) \circ \dots \circ \mathrm{dec}(\hat X_{\mathsf{blocks},1}, \dots, \hat X_{\mathsf{blocks},\mathsf{vars}_2})$ is the Boolean assignment encoded by $\hat X$, and we have $Q_{i,j}(\hat X) \in \{0,1\}$ as well as $Q_{i,j}(\hat X) = 0$ if and only if $\mathrm{dec}(\hat X) \models \ell_{i,j}$.
    Finally, for any $i \in [m]$, define the polynomial
    \[P_i(X) = \prod_{j=1}^k Q_{i,j}(X).\]
    Now for all $\hat X \in \field^{\mathsf{blocks}\cdot \mathsf{vars}_2}$, we have that $P_i(\hat X) \in \{0,1\}$ and $P_i(\hat X) = 0$ if and only if $\mathrm{dec}(\hat X) \models C_i$.
    Consequently, the polynomial equation system $E \coloneqq \{P_i\}_{i \in [m]}$ has a solution if and only if $\FF$ is satisfiable.

    We now verify that $E$ has the claimed properties and can be constructed in the claimed running time.
    Note that by definition and the fact that $k$, $q$, and $\delta$ are considered to be constant, $\mathsf{vars}_1$ and $\mathsf{vars}_2$ are constant, and $\mathsf{blocks}$ is linear in $n$.
    By construction, $E$ consists of exactly $m$ polynomial equations and uses at most $\mathsf{blocks} \cdot \mathsf{vars}_2$ variables.
    This directly yields the desired bound:
    \begin{align*}
        \mathsf{blocks} \cdot \mathsf{vars}_2 & = \left\lceil \frac{n}{\mathsf{vars}_1} \right\rceil \cdot \left\lceil \frac{\mathsf{vars}_1}{\log q} \right\rceil        \\
                                              & \leq \left( \frac{n}{\mathsf{vars}_1} + 1 \right) \cdot \left( \frac{\mathsf{vars}_1}{\log q} + 1 \right)                 \\
                                              & = \frac{n}{\log q} + \frac{n}{\mathsf{vars}_1} + \frac{\mathsf{vars}_1}{\log q} + 1                                       \\
                                              & = \frac{n}{\log q} \cdot \left( 1 + \frac{\log q}{\mathsf{vars}_1} + \frac{\mathsf{vars}_1}{n} + \frac{\log q}{n} \right) \\
                                              & \in \frac{n}{\log q} \cdot \left( 1 + \frac{\delta}{2} + o(1) \right).
    \end{align*}
    The degree of any polynomial $\mathrm{DEC}_{v_1}$ for $v_1 \in [\mathsf{vars}_1]$ is trivially bounded by $\mathsf{vars}_2 \cdot (q-1)$, as it is a $\mathsf{vars}_2$-variate polynomial over $\field$.
    The same is true for the polynomial $Q_{i,j}$ for any $i \in [m]$ and $j \in  [k]$ (the remaining variables are unused in $Q_{i,j}$).
    Consequently, $P_i$ is of degree at most $k \cdot \mathsf{vars}_2 \cdot (q-1) \leq k\cdot(\frac{2}{\delta} + \frac{1}{\log q} + 1) \cdot (q-1) \in \OO(1)$ for all $i \in [m]$.

    Finally, we analyze the running time of the reduction.
    For any $i \in [m]$, the polynomial $P_i$ can be viewed as a polynomial with at most $k \cdot \mathsf{vars}_2$ variables, and can be interpolated from all of its evaluations on arbitrary assignments to those variables using Lemma~\ref{lem:interpolation}.
    To obtain these evaluations, we compute all relevant evaluations of $Q_{i,j}$ for all $j \in [k]$, that is, all evaluations for arbitrary assignments to the variables occurring in $Q_{i,j}$ but some fixed assignment for the variables not occurring.
    Any evaluation of $P_i$ can then be computed as the product of the corresponding evaluations of the $Q_{i,j}$.

    To obtain the evaluations of $Q_{i,j}$ for all $i \in [m]$ and $j \in [k]$, we first need to compute all evaluations of $\mathrm{DEC}_{v_1}$ for all $v_1 \in [\mathsf{vars}_1]$.
    Using the definition of $\mathrm{dec}$ suggested above, this can be done as follows:
    Simultaneously count from $0$ to $q^{\mathsf{vars}_2}-1$ in base $q$ using variables $q_0, \dots, q_{\mathsf{vars}_2} \in \field$ and in base $2$ (modulo $2^{\mathsf{vars}_1}$) using variables $d_0, \dots, d_{\mathsf{vars}_1} \in \{0,1\}$ and for each step, take $d_{v_1}$ as the evaluation of $\mathrm{DEC}_{v_1}$ on input $(q_0, \dots, q_{\mathsf{vars}_2})$ for all $v_1 \in [\mathsf{vars}_1]$.
    This takes time $\OO(\mathsf{vars}_1 \cdot q^{\mathsf{vars}_2})$.

    The polynomial $Q_{i,j}$ is now treated as a $\mathsf{vars}_2$-variate polynomial (leaving out all unused variables).
    Each evaluation of $Q_{i,j}$ is now simply obtained by a lookup in the correct evaluation of the correct polynomial $\mathrm{DEC}_{v_1}$ and potentially a negation.
    Hence, all evaluations of a single $Q_{i,j}$ can be obtained in time $\OO(q^{\mathsf{vars}_2})$.
    (Here, we assume that a single literal $\ell_{i,j}$ can be read from the input in constant time.
    If this is not possible a priori, we can prepare a data structure for it with linear pre-processing.)
    In consequence, all evaluations of all polynomials~$Q_{i,j}$ for $i \in [m]$ and $j \in [k]$ can be obtained in time $\OO(mk \cdot q^{\mathsf{vars}_2})$.

    Now, for any $i \in [m]$, the evaluation of $P_i$ for a single assignment $X$ for the variables occurring in $P_i$ can be computed in time $k$ by computing the product $\prod_{j=1}^k Q_{i,j}(X)$.
    (Here, $Q_{i,j}$ is treated as a polynomial with the same variables as $P_i$, so one has to project the assignment $X$ to the correct subset of variables for the lookup.)
    The polynomial $P_i$ can now be interpolated in time $\OO(\mathsf{blocks} \cdot \mathsf{vars}_2 \cdot q^{k \cdot \mathsf{vars}_2 \cdot (q-1)})$ using Lemma~\ref{lem:interpolation}.
    Hence, this step takes time $\OO(m \cdot \mathsf{blocks} \cdot \mathsf{vars}_2 \cdot q^{k \cdot \mathsf{vars}_2 \cdot (q-1)})$ for the whole system $E = (P_i)_{i \in [m]}$.

    In total, this means that the running time is bounded by
    \[\OO(\mathsf{vars}_1 \cdot q^{\mathsf{vars}_2} + mk \cdot q^{\mathsf{vars}_2} + m \cdot \mathsf{blocks} \cdot \mathsf{vars}_2 \cdot q^{k \cdot \mathsf{vars}_2 \cdot (q-1)}) \subseteq \OO(nm).\]

    Finally, the above reduction can be strengthened to a parsimonious reduction by enforcing a bound on the possible assignments to any block of variables over \field and ensuring that the decoding function restricted to the possible assignments is a bijection.
    With this modification, there is a $1$-to-$1$ correspondence between satisfying assignments of the original formula and solutions for the constructed polynomial equation system.

    This bound can be realized by constructing a $\mathsf{vars}_2$-variate polynomial $B$ over \field that expresses that the number represented by the given block when interpreted as a base-$q$ number is less than $2^{\mathsf{vars}_1}$.
    For this, use the same approach as for constructing the polynomials $\mathrm{DEC}_{\mathsf{vars}_1}$:
    List all of its evaluations and interpolate using Lemma~\ref{lem:interpolation}.
    Finally, let $\mathrm{BOUND}_{b}(X)$ be the composition of $B$ with a projection to the $b$-th block of $X$.
    Now the desired polynomial equation system is $\{P_1, \dots, P_m, \mathrm{BOUND}_1, \dots, \mathrm{BOUND}_{\mathsf{blocks}}\}$.
\end{proof}

\subsection{Lower Bounds for \texorpdfstring{\PESqd{q}{d}}{PES}.}

We are now ready to prove Theorem~\ref{thm:pes-hardness} as a direct corollary to Lemma~\ref{lem:PESreduction}.

\PEShardness*

\begin{proof}
    We prove the contrapositive.
    Assume that there exists a prime power~$q$ and rational~$\delta>0$ such that, for all $d\in\N$, there is an $\OO(q^{(1-\delta)n})$-time algorithm for \PESqd{q}{d}.
    We show that this contradicts \pp{SETH}.

    Let $\eps = \delta/2$, and let $k\ge 3$ be any integer.
    We now devise an $\OO(2^{(1-\eps)n})$-time algorithm for~\pp{$k$-SAT}, contradicting \pp{SETH}.
    The algorithm is given any $n$-variable $k$-CNF formula~$\FF$, uses Lemma~\ref{lem:PESreduction} to reduce it to an equivalent instance~$E$ of \PESqd{q}{d} for $d = k \cdot \left\lceil \frac{\lceil \frac{2}{\delta} \cdot \log q \rceil}{\log q} \right\rceil \cdot (q-1)$, and finally applies the assumed algorithm for \PESqd{q}{d}.
    As $k$, $q$ and $\delta$ are constant in this setting, the reduction takes time $\OO(nm)$, where~$m$ is the number of clauses in $\FF$, by Lemma~\ref{lem:PESreduction}.
    Furthermore, the polynomial equation system~$E$ has at most $\frac{n}{\log q} \cdot \left( 1 + \frac{\delta}{2} + o(1) \right)$ variables.
    In consequence, the final application of the assumed algorithm for \PESqd{q}{d} takes time
    \[\OO\left( q^{(1-\delta) \cdot \frac{n}{\log q} \cdot \left( 1 + \frac{\delta}{2} + o(1) \right)} \right)
        \le \OO\left( 2^{(1-\delta) \cdot \left( 1 + \frac{\delta}{2} + o(1) \right) \cdot n} \right).\]
    To finish the proof, we show that this is in $\OO(2^{(1-\eps)n})$.
    Indeed, we have:
    \begin{align*}
        (1-\delta) \cdot (1 + \frac{\delta}{2} + o(1) ) & = 1 + \frac{\delta}{2} + o(1) - \delta - \delta \cdot \left( \frac{\delta}{2} + o(1) \right) \\
                                                        & = 1 - \delta \cdot \left( \frac{1}{2} + \frac{\delta}{2} + o(1) \right) + o(1).
    \end{align*}
    Now for $\eps' \coloneqq \delta \cdot \left( \frac{1}{2} + \frac{\delta}{2} + o(1) \right)$ we have $\eps < \eps'$, which means that the running time is
    \[\OO(2^{(1- \eps' + o(1)) n}) \le \OO(2^{(1-\eps)n}).\]
    This concludes the proof of Theorem~\ref{thm:pes-hardness}.
\end{proof}

\subsection{Lower Bound for Counting the Roots of a Polynomial.}

Lemma~\ref{lem:PESreduction} is a mapping reduction that is \emph{parsimonious}, that is, it preserves the number of solutions.
Since the proof of Theorem~\ref{thm:pes-hardness} only deals with the parameters of Lemma~\ref{lem:PESreduction} and the running times, the theorem can be safely lifted to its counting version.
That is, if \#{SETH} holds, then \#\PESqd{q}{d} cannot be solved fast.
Here, \#{SETH} is the counting variant of \pp{SETH}, which states that counting the number of satisfying assignments to $k$-CNF formulas cannot be done fast, and \#\PESqd{q}{d} is the problem of computing the number of solutions to an instance of \PESqd{q}{d}.

If the number of equations in a given instance of \#\PESqd{q}{d} is~$1$, then the problem is simply to compute the number of roots of the one given polynomial.
For a polynomial $P \in \field[X_1,\dots,X_n]$, any vector $x\in\field^n$ with $P(x)=0$ is called a \emph{root} of the polynomial.
For convenience, we define the problem \Roots{q}{d} as the special case of \#\PESqd{q}{d} with $m=1$:
\defproblemtalt{\Roots{q}{d}}{Polynomial $P \in \field[X_1,\dots,X_n]$ of degree at most~$d$}{How many roots does $P$ have?}

In a beautiful paper and somewhat surprisingly, Williams~\cite{DBLP:conf/soda/Williams18a} was able to reduce from \#\PESqd{q}{d} to \Roots{q}{D} for some $D\ge d$.
The following reduction is implicit in~\cite[proof of Theorem~4]{DBLP:conf/soda/Williams18a}.

\begin{lemma}[Williams~{\cite{DBLP:conf/soda/Williams18a}}]%
\label{lem:williams-roots-reduction}
    Let $q$ be a prime power and $d, B, C \in\N$.
    For every $\eps>0$, there is a constant~$D\in\N$ and an oracle reduction from $\#\PESqd{q}{d}$ to $\Roots{q}{D}$, such that on input~$P_1,\dots,P_m\in\field[X_1,\dots,X_n]$, the reduction makes queries to polynomials $Z\in\field[X_1,\dots,X_n]$ and satisfies the following property:
    If every polynomial $P_i$ only depends on at most $B$ variables, then these queries~$Z$ have degree at most $D\coloneqq BC \cdot (q-1)$ and the reduction runs in time $\OO(q^{m/C}\cdot (m + n))$.
\end{lemma}
Williams~\cite{DBLP:conf/soda/Williams18a} used Lemma~\ref{lem:williams-roots-reduction} to show that $\Roots{q}{d}$ does not have algorithms that run in time~$\OO(2^{(1-\delta) n})$, unless \pp{\#SETH} is false.
This hardness result is only tight for~$q=2$, and we extend it to a tight hardness result for arbitrary prime powers~$q$.
\ROOTShardness*
In order to prove Theorem~\ref{thm:roots-hardness}, we follow the same plan as Williams' proof of his \pp{\#SETH} lower bound: we reduce from \pp{\#$k$-SAT} via \#\PESqd{q}{d} to \Roots{q}{D}, but we replace the first reduction with Lemma~\ref{lem:PESreduction} as a key missing component.
More precisely, starting from \pp{\#$k$-SAT}, we apply the counting version of the sparsification lemma~\cite{DBLP:journals/jcss/ImpagliazzoPZ01,DBLP:conf/coco/CalabroIP06,DBLP:journals/talg/DellHMTW14} to make the number~$m$ of clauses at most $(k/\eps)^{\OO(k)}$.
Next, we apply our efficient reduction of Lemma~\ref{lem:PESreduction} from \#\pp{$k$-SAT} to \#\PESqd{q}{d}; the main feature of this reduction is that it compresses $n$ Boolean variables from the $k$-CNF formula into roughly $n/\log q$ variables over $\field$ in the polynomial equation system.
Finally, we apply Lemma~\ref{lem:williams-roots-reduction} to reduce the polynomial equation system to \Roots{q}{D}; we apply this reduction with parameters satisfying $m/C = \eps n'$, where $n'$ is the number of variables in the original \#\pp{$k$-SAT} formula.
This leads to a running time of $\OO(q^{\eps n'} \cdot (n+n'))$, which is sufficient to obtain the desired lower bound under \#\pp{SETH}.

\begin{proof}(Proof of Theorem~\ref{thm:roots-hardness})
    We follow the proof outline stated above, so suppose there is some $\delta>0$ such that \Roots{q}{d} can be solved in time $\OO(q^{(1-\delta)n})$ for all $d$.
    We now devise an algorithm for \#\pp{$k$-SAT} that contradicts \pp{\#SETH}, using a constant $\eps > 0$ that we will set later.
    Let $\varphi$ be a $k$-CNF with $n$ variables.
    Without loss of generality, we can assume that $\varphi$ has at most $(k/\eps)^{\OO(k)} \cdot n$ clauses by the counting version of the sparsification lemma~\cite{DBLP:journals/jcss/ImpagliazzoPZ01,DBLP:conf/coco/CalabroIP06,DBLP:journals/talg/DellHMTW14}.
    This causes a multiplicative $\OO(2^{\eps n})$ overhead in the running time.
    We now apply Lemma~\ref{lem:PESreduction} to obtain a polynomial equation system $G$ using at most $\frac{n}{\log q}\cdot \left(1 + \frac{\delta}{2} + o(1)\right)$ variables and $(k/\eps)^{\OO(k)}\cdot n + \left\lceil \frac{n}{\lceil \delta/2 \cdot \log q \rceil}\right\rceil \in (k/\eps)^{\OO(k)} \cdot n$ equations, where the number of satisfying assignments of $\varphi$ is exactly the number of solutions to $G$.
    This takes time $\OO((k/\eps)^{\OO(k)} \cdot n^2)$.
    From the proof of Lemma~\ref{lem:PESreduction}, we can also see that each equation in $G$ uses at most $B \coloneqq \lceil \lceil 2/\delta \cdot \log q \rceil / \log q \rceil \cdot k$ variables.
    Let $m$ be the number of equations in $G$, $d \coloneqq B \cdot (q-1)$, and $C \coloneqq m/(\eps n)$.

    Finally, apply Lemma~\ref{lem:williams-roots-reduction} to $G$ with $q$, $d$, $B$, $C$, and $\eps$ as defined above.
    This takes time
    \[\OO\left( q^{m/(m/(\eps n))} \cdot \left( m + \left(\frac{n}{\log q} \cdot (1 + \delta/2 + o(1)\right) \right) \right) = \OO(q^{\eps n} \cdot n)\]
    and any polynomial queried in the reduction uses at most $\frac{n}{\log q} \cdot (1 + \frac{\delta}{2} + o(1))$ variables and has degree at most $D \coloneqq BC \cdot (q-1)$.

    Apart from this, we get the following running time for computing the answers to all oracle queries using the assumed algorithm for \Roots{q}{D}:
    \begin{align*}
      \OO(2^{\eps n} \cdot q^{\eps n} \cdot q^{(1-\delta)\cdot \frac{n}{\log q} \cdot (1 + \delta/2 + o(1))}) &\leq \OO(q^{(1-\delta)\cdot \frac{n}{\log q} \cdot (1 + \delta/2 + o(1)) + 2\eps n}).
    \end{align*}
    Setting $\delta'$ to a constant with $\delta' < \delta \cdot (1/2 + \delta/2 + o(1))$, by the same calculations as in the proof of Theorem~\ref{thm:pes-hardness}, the above is bounded by
    \[q^{(1-\delta')n + 2\eps n} = q^{(1 + 2\eps - \delta') n}.\]
    This term obviously dominates the running time of our algorithm.
    Setting $\eps < \delta'/2$, this implies an algorithm for \#\pp{$k$-SAT} running in time $\OO(q^{(1 - \delta'')n})$ for $\delta'' = \delta' - 2\eps$, yielding the desired \#\pp{SETH} lower bound.
\end{proof}

\bibliography{refs}

\end{document}